%% file: sample-sigconf.tex
\documentclass[sigconf]{acmart}
\usepackage{booktabs} 

\usepackage{multirow}
\usepackage{multicol}
\usepackage{siunitx}
\usepackage{comment}
\usepackage{balance}  
\usepackage{graphicx}
\usepackage{float}
\usepackage{subcaption}
\usepackage{xcolor}
\usepackage{makecell}
\usepackage{tikz}
\usepackage{lipsum}

\newcommand*{\vaff}{V_\textsc{aff}}
\newcommand*{\vcpaff}{V_\textsc{aff+}}

\newcommand*{\dlm}{d^L_G}
\newcommand*{\dlmp}{d^L_{G'}}

\newcommand*{\dmax}{\beta}
\newcommand*{\lml}[1]{|{#1}|_\textsc{l}}
\newcommand*{\lmf}{l}
\newcommand*{\delf}{e}
\newcommand*{\true}{\text{\normalfont True}}
\newcommand*{\false}{\text{\normalfont False}}

\tikzstyle{vertex}=[circle,fill=black!10,minimum size=12pt,inner sep=0pt]
\tikzstyle{edge} = [draw,thick]
\tikzset{top/.style={baseline=(current bounding box.north)}}
\tikzset{mid/.style={baseline=(current bounding box.center)}}
\newcommand{\tikzcaption}[1]{\node[below=4mm] at (current bounding box.base) {#1}}

\usepackage{amsmath,amsthm}
\usepackage[ruled,vlined,linesnumbered]{algorithm2e}

\newtheorem{theorem}{Theorem}[section]
\newtheorem{lemma}[theorem]{Lemma}
\newtheorem{example}[theorem]{Example}

\newtheorem{definition}[theorem]{Definition}

\usepackage{enumitem}

\newcommand*{\n}[1]{\textcolor{black}{#1}}

\AtBeginDocument{%
  \providecommand\BibTeX{{%
    \normalfont B\kern-0.5em{\scshape i\kern-0.25em b}\kern-0.8em\TeX}}}

\copyrightyear{2022} \acmYear{2022} \setcopyright{acmlicensed}\acmConference[SIGMOD '22]{Proceedings of the2022 International Conference on Management of Data}{June 12--17,2022}{Philadelphia, PA, USA}\acmBooktitle{Proceedings of the 2022 International Conference onManagement of Data (SIGMOD '22), June 12--17, 2022, Philadelphia, PA, USA}\acmPrice{15.00}\acmDOI{10.1145/3514221.3517883}\acmISBN{978-1-4503-9249-5/22/06}

\begin{document}
\fancyhead{}

\title{BatchHL: Answering Distance Queries on Batch-Dynamic Networks at Scale}


\author{Muhammad Farhan}
\affiliation{%
  \institution{Australian National University}
  \city{Canberra}
  \country{Australia}}
\email{muhammad.farhan@anu.edu.au}

\author{Qing Wang}
\affiliation{%
  \institution{Australian National University}
  \streetaddress{1 Th{\o}rv{\"a}ld Circle}
  \city{Canberra}
  \country{Australia}}
\email{qing.wang@anu.edu.au}

\author{Henning Koehler}
\affiliation{%
  \institution{Massey University}
  \streetaddress{1 Th{\o}rv{\"a}ld Circle}
  \city{Palmerston North}
  \country{New Zealand}}
\email{H.Koehler@massey.ac.nz}

\begin{abstract}
Many real-world applications operate on dynamic graphs that undergo rapid changes in their topological structure over time. However, \n{it is challenging to design dynamic algorithms that are capable of supporting such graph changes efficiently. To circumvent the challenge, we propose a batch-dynamic framework for answering distance queries, which combines offline labelling and online searching to leverage the advantages from both sides - accelerating query processing through a partial distance labelling that is of limited size but provides a good approximation to bound online searches. We devise batch-dynamic algorithms to dynamize a distance labelling efficiently in order to reflect batch updates on the underlying graph.} In addition to providing theoretical analysis for the correctness, labelling minimality, and \n{computational complexity}, we have conducted experiments on \n{14 real-world networks to empirically verify the efficiency and scalability of the proposed algorithms.}
\end{abstract}

\ccsdesc[500]{Theory of computation}
\ccsdesc[500]{Theory of computation~Data structures and algorithms for data management}

\keywords{Shortest-path distance; batch-dynamic graphs; 2-hop cover; highway cover; distance labelling maintenance; graph algorithms}


\maketitle

\input{section_Introduction}
\input{section_RelatedWork}
\input{section_Preliminaries}
\input{section_Algorithms}
\input{section_ParallelLabelling}
\input{section_Experiments}
\input{section_Conclusion}

\bibliographystyle{ACM-Reference-Format}
\balance
\bibliography{sample-base}

\end{document}

%% file: section_Introduction.tex
\section{Introduction}

Graphs in real-world applications are typically dynamic, undergoing discrete changes in their topological structure by either adding or deleting edges and vertices. However, due to the rapid nature of data acquisition, it is often unrealistic to process single changes sequentially on graphs. Rather, updates may be aggregated in batches, and graphs are updated by large batches of updates \cite{dhulipala2020parallel}.

\vspace{0.1cm}
\noindent\textbf{Applications.~}There are various real-world applications operating on graphs that undergo rapid changes \cite{ukkonen2008searching,vieira2007efficient,backstrom2006group,boccaletti2006complex}, such as communication networks, context-aware search in web graphs \cite{potamias2009fast,ukkonen2008searching}, social network analysis in social networks \cite{vieira2007efficient,backstrom2006group}, route-planning in road networks \cite{abraham2012hierarchical,delling2014robust}, management of resources in computer networks \cite{boccaletti2006complex}, and so on.
We discuss a few examples below.
\begin{itemize}[leftmargin=*]
    \item In communication networks, links between network devices (e.g. routers) may become slow or broken due to congestion of information flow over a network or a deadly fault in a network device. Efficient maintenance of shortest paths to reflect the underlying changes helps vendors to activate new links and preserve the quality of their service \cite{boccaletti2006complex}.
    \item In social networks, Twitter is highly dynamic~\cite{myers2014bursty} -- about 9\% of all connections change in a month. Users having 100 followers on average were found to obtain 10\% more new followers but lose about 3\% of existing followers in a given month. Distance information is often used to recommend the relevant content or new connections~\cite{yahia2008efficient,vieira2007efficient}.
\end{itemize}

Although the batch-dynamic setting is increasingly important and desired for real-world applications, it poses significant challenges on algorithm design due to the combinatorial explosion of different interactions possibly occurring among updates.
Very recently, several batch-dynamic algorithms have been reported, mostly focusing on traditional graph problems such as graph connectivity \cite{acar2019parallel}, dynamic trees \cite{acar2020parallel} and k-clique counting \cite{dhulipala2020parallel-clique}. As of yet, batch-dynamic algorithms for shortest-path distance have been left unexplored, despite the fact that 
computing the distance between an arbitrary pair of vertices (i.e., \emph{distance queries}) is a fundamental problem in many real-world applications.
Up to now, only several dynamic labelling algorithms for distance queries have been studied in the single-update setting, which handles one single update (e.g., edge insertion or edge deletion) at a time \cite{akiba2014dynamic,d2019fully,qin2017efficient,hayashi2016fully}. Unlike previous studies, in this work, we are interested in exploring the following research questions:
\begin{itemize}[leftmargin=*]
    \item[--] Is it possible to design batch-dynamic algorithms for distance queries, which can efficiently reflect batch updates on graphs? 
    \item[--] Can such batch-dynamic algorithms offer significant performance gains in comparison with state-of-the-art algorithms in the single-update setting? 
    \item[--] Can we parallelize such batch-dynamic algorithms to further boost performance in a parallel setting, whenever parallel computing resources are available?
\end{itemize}
\begin{figure*}[ht]
\centering
\includegraphics[width=0.99\textwidth]{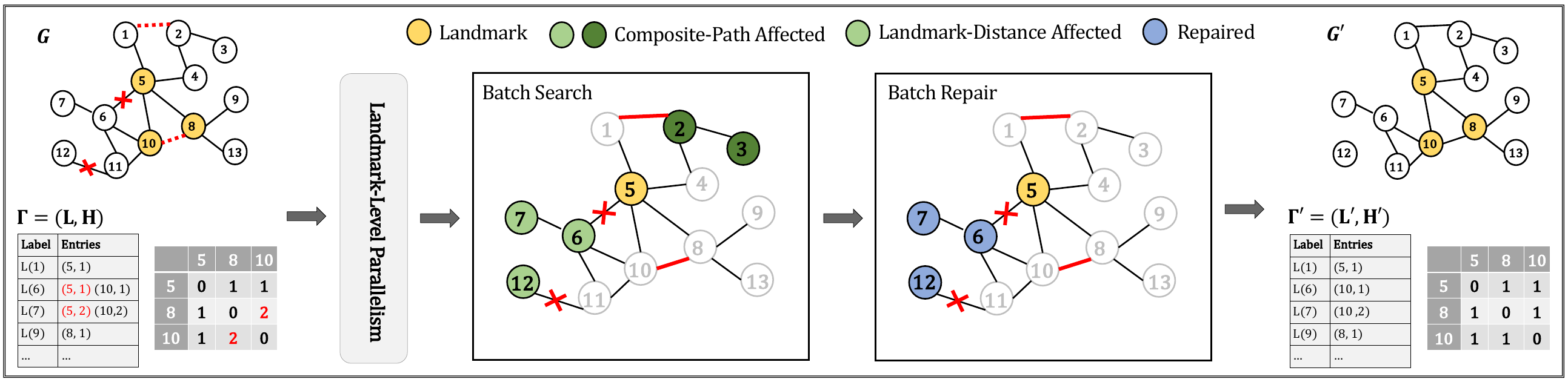}
\vspace*{-0.3cm}
\caption{A high-level overview of our batch-dynamic method (BatchHL) which performs a batch update in two phases: 1) Batch Search: find vertices that are affected, and 2) Batch Repair: repair vertices returned by Batch Search.}\label{fig:overview}
\label{fig:intro}\vspace*{-0.4cm}
\end{figure*}
\noindent\textbf{Present work.~}\n{The goal of this work is to answer the aforementioned research questions on complex networks (e.g., social networks and web graphs). It is known that complex networks exhibit different properties (e.g., small diameter) from road networks \cite{akiba2013fast}. Specifically, we thus aim to develop an efficient and scalable solution for answering distance queries on complex networks that undergo batch updates.} This solution should have the following properties: (1) \emph{Time efficiency}: it can quickly answer distance queries in a way that reflects batch updates on graphs, e.g., microseconds for large-scale graphs, since distances are used as a building block in many graph analysis tasks; (2) \emph{Space efficiency}: it can use an efficient data structure for storing distance labellings, which ideally grows linearly or sublinearly with the number of vertices in a graph; (3) \emph{Scalability}: it can scale to large graphs with millions of vertices and edges without compromising query and update performance. 

To derive these properties, firstly, we choose to combine offline labelling and online searching so as to leverage the advantages from both sides - accelerating query processing through \n{a \emph{partial distance labelling} that is of limited size but provides a good approximation} to bound online searches. Traditional labelling methods such as pruned landmark labelling (PLL) \cite{akiba2013fast} can efficiently answer distance queries \n{using a \emph{full distance labelling};} however, their labelling size grows quadratically with the size of a graph and the computational cost of updating such labellings to reflect rapid changes is often unbearably high. Hence, we consider to \n{use a partial distance labelling} for providing an upper bound for online search (with theoretical guarantees for exact answers, which will be discussed in Section \ref{sec:BDL}). This brings two significant computational benefits: (i) labelling construction can scale to very large graphs; (ii) labelling maintenance can be efficiently handled on dynamic graphs.

We propose a batch-dynamic method, BatchHL, to dynamize distance labellings efficiently in order to reflect large batches of updates on a graph. BatchHL consists of two phases: (1) \emph{Batch search} finds vertices whose labels are affected by batch updates; (2) \emph{Batch repair} changes the labels of affected vertices to ensure correctness and minimality of labelling. The main challenges are  the following:
\begin{itemize}[leftmargin=*]
    \item[--] \emph{Unifying edge insertion and deletion:} We explore the core properties shared by edge insertion and edge deletion. Based on this, we unearth an elegant pattern that unifies these two fundamental kinds of graph updates.
    \item[--] \emph{Avoiding unnecessary and repeated computation:} We analyse how updates interact with each other, and based on that, design pruning rules to reduce search and repair spaces so as to leverage the computational efficiency of batch updates.  
   \item[--] \emph{Exploiting the potential of parallelism:} We consider to parallelize batch search and batch repair in a simple but easy-to-implement way to speedup the performance.
\end{itemize}
To the best of our knowledge, this is the first study to develop a batch-dynamic solution for answering distance queries on large-scale graphs. Figure~\ref{fig:overview} presents the high-level overview of BatchHL which performs batch search and then batch repair. Figure~\ref{fig:affected-vertices} shows the gaps in the number of vertices affected by batch updates when different variants of our method are used in the batch-update setting, in comparison with the single-update setting. Notice that the number of affected vertices in the single-update setting (i.e., UHL) is much higher than the ones in the batch-update setting (i.e., BHL$^s$, BHL and BHL$^+$). This is because one vertex may be affected by multiple updates in a batch, which would unavoidably lead to repeated and unnecessary computations in the single-update setting.

\smallskip
\noindent{\textbf{Contributions.}} The main contributions of this work are as follows:
\begin{itemize}[leftmargin=*]
\item We propose a batch-dynamic method which can handle batch updates efficiently and uniformly so as to reflect batch updates on graphs into a highway cover labelling. Previous studies \cite{akiba2014dynamic,d2019fully} reported that handling edge deletions on a graph has been recognized as being computational expensive and difficult, even in the single-update setting. Our method alleviates this challenge and can handle both edge insertions and edge deletions in batches efficiently.\looseness=-1
\item We develop efficient pruning strategies in our method, i.e., in both batch search and batch repair, to eliminate repeated and unnecessary computations on graphs.
As a result, when dealing with batch updates, we traverse much smaller numbers of vertices than in the single-update setting where each update is handled independently. We also design an inference mechanism to compute new distances based on boundary vertices and incorporate this into batch repair in our method.
\item We prove that the proposed method can preserve the minimality of labelling on batch-dynamic graphs. Notice that, as discussed in \cite{d2019fully}, minimality is a difficult but highly desirable property to have for designing a distance labelling over dynamic graphs. Otherwise, a distance labelling may have increasingly unnecessary entries left in its labels and query performance would deteriorate over time.\looseness=-1
\item Our proposed method can scale to very large dynamic graphs. This is due to several reasons: the design choices on combining offline labelling and online searching, the properties of highway cover labelling, the pruning strategies in batch search and batch repair, and landmark-based parallelism. We will discuss these in detail in Section \ref{sec:experiments}. 
\end{itemize}
We have evaluated our method on \n{14 real-world networks} to verify their efficiency and scalability. The results show that our method significantly improves both time and space efficiency compared to the state-of-the-art methods. It can maintain a very small labelling size, while still answering queries in the order of milliseconds, even on large-scale graphs with several billions of edges that undergo large batch updates. For example, the average distance query time for the UK dataset with 3.7 billions of edges is around 1 millisecond, and the average update time for each update is 14.45 milliseconds; \n{similarly, the average distance query time for Twitter dataset is 0.86 millisecond and the average update time is 13.29 milliseconds, more than 300 times faster than the state-of-the-art method.}

\begin{figure}[t!]
\centering
\includegraphics[scale=0.47]{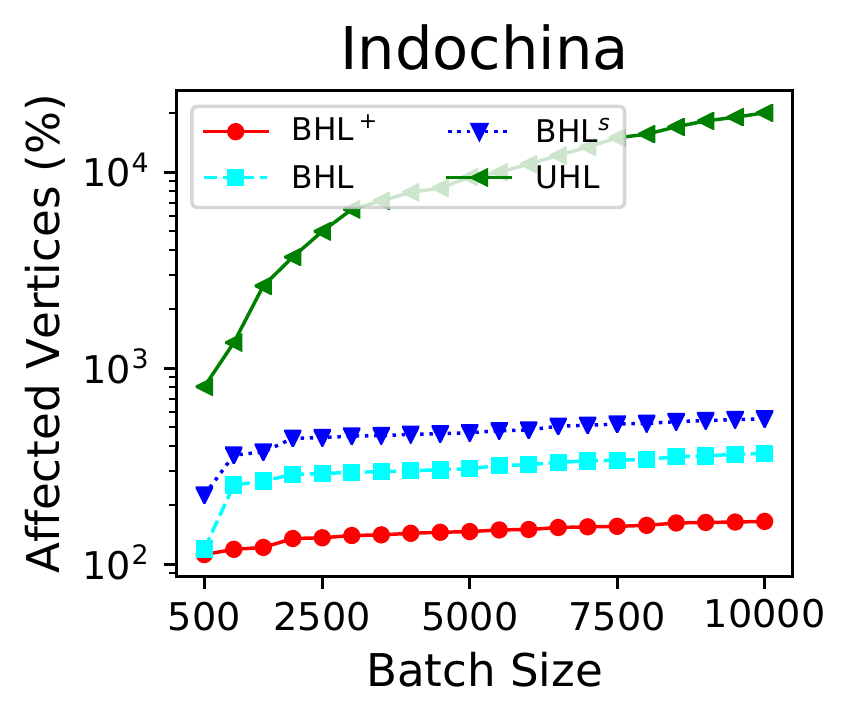}
\includegraphics[scale=0.47]{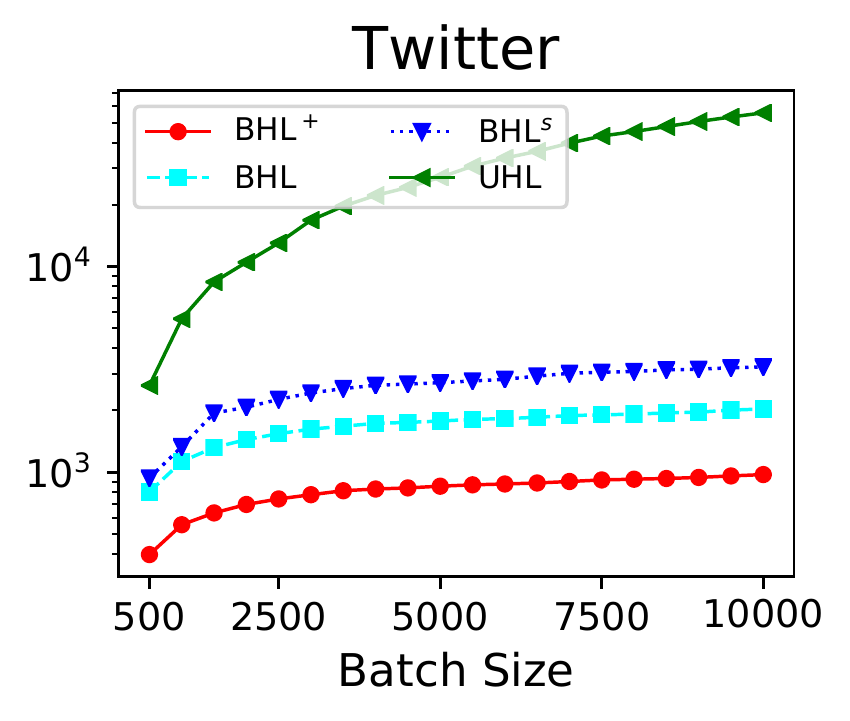}\vspace{-0.3cm}
\caption{\n{Number of vertices affected by batch updates of varying sizes. BHL and BHL$^+$ are our batch dynamic algorithms, BHL$^s$ is a variant of BHL which splits edge insertions and deletions into sub-batches and performs them sequentially, and UHL handles updates in the single-update setting.}}\label{fig:affected-vertices}\vspace*{-0.4cm}
\label{fig:intro}
\end{figure}

%% file: section_RelatedWork.tex
\section{RELATED WORK}\label{sec:background}
Answering distance queries has been an active research topic for many years. Traditionally, distance queries can be answered using Dijkstra's search on positively weighted graphs or a breadth-first search (BFS) on unweighted graphs \cite{tarjan1983data} or a bidirectional scheme combining two such searches: one from the source vertex and the other from the destination vertex \cite{pohl1971bi}. However, these algorithms may traverse an entire network when two query vertices are far apart from each other and become too slow for large networks. To accelerate response time in answering distance queries, labelling-based methods have emerged as an attractive way, which precompute a data structure, called \emph{distance labelling} \cite{cohen2003reachability,akiba2013fast,akiba2012shortest,fu2013label,jin2012highway,delling2014robust,abraham2012hierarchical,wei2010tedi,farhan2018highly,li2017experimental,chang2012exact,wang2021query}. For example, Akiba et al. \cite{akiba2013fast} proposed the pruned landmark labelling (PLL) to pre-compute a 2-hop distance labelling \cite{cohen2003reachability} by performing a pruned breadth-first search from every vertex, \n{and Li et al. \cite{li2019scaling} developed a parallel algorithm for constructing PLL which achieved the state-of-the-art results for answering distance queries on static graphs.}

So far, several attempts have been made to study distance queries over dynamic graphs \cite{akiba2014dynamic,qin2017efficient,hayashi2016fully,d2019fully,ouyang2020efficient,zhang2021efficient,farhan2021fast,farhan2021efficient} which only considered the unit-update setting i.e., to perform updates one at a time. Akiba et al. \cite{akiba2014dynamic} studied the problem of updating PLL for incremental updates (i.e. edge additions). This work however does not remove outdated entries because the authors considered it too costly. Qin et al. \cite{qin2017efficient} and D'angelo et al. \cite{d2019fully} studied the problem of updating PLL for decremental updates (i.e. edge deletions). Note that, in the decremental case, outdated distance entries have to be removed; otherwise distance queries cannot be correctly answered. Their methods suffer from high time complexities and cannot scale to large graphs, e.g., the average update time of an edge deletion on a network with 19M edges is 135 seconds in \cite{qin2017efficient} and on a network with 16M edges is 19 seconds in \cite{d2019fully}. D'angelo et al. \cite{d2019fully} combined the algorithm for incremental updates proposed in \cite{akiba2014dynamic} with their method for decremental updates to form a fully dynamic algorithm, which can only be applied to networks with around 20M of edges. Hayashi et al. \cite{hayashi2016fully} proposed a fully dynamic method which combines a distance labelling with online search to answer distance queries. Their method pre-computes bit-parallel shortest-path trees (SPTs) rooted at each $r \in R$ for a small subset of vertices $R$ and dynamically maintain the correctness of these bit-parallel SPTs for every edge insertion and deletion. Then, an online search is performed under an upper distance bound computed via the bit-parallel SPTs on a sparsified graph. \n{Unlike these approaches, our present work considers the batch-update setting.} Designing dynamic algorithms is usually quite complex and difficult, as reported by these approaches in the single-update setting, and arguably even more so in the batch-dynamic setting or parallel setting.

Another line of research studied streaming graph algorithms. In the streaming setting, a rapidly changing graph is often modeled using certain compressed data structures due to space constraints. Updates are received as a stream, but may be accumulated into batches through a sliding window and applied to the underlying graph. In this setting, a number of methods \cite{mcgregor2014graph,feigenbaum2005graph,pacaci2020regular} have been proposed to address distance queries. However, these methods operate under certain constraints, e.g., limited amount of memory and accuracy of graph structure. Different from these streaming graph methods, our work considers applications which operate on batch-dynamic graphs that are explicitly stored and can be processed in the main memory of a single machine. Nevertheless, the ideas of our algorithm can be easily extended to deal with batch updates in the streaming setting. \vspace{-0.1cm}

%% file: section_Preliminaries.tex
\section{Preliminaries}\label{sec:preliminaries}
Let $G = (V, E)$ be a graph where $V$ is a set of vertices and $E \subseteq V \times V$ is a set of edges. The \emph{distance} between two vertices $s$ and $t$ in $G$, denoted as $d_G(s, t)$, is the length of a shortest path between $s$ and $t$. If there does not exist any path between $s$ and $t$, then $d_G(s, t) = \infty$. We use $P_{G}(s,t)$ to denote the set of all shortest paths between $s$ and $t$ in $G$, and $N(v)$ the set of neighbors of a vertex $v \in V$, i.e. $N(v) = \{v' \in V | (v, v') \in E \}$. Without loss of generality, we focus our discussion on unweighted, undirected graphs in \n{this paper} and discuss the extension to directed and \n{non-negative} weighted graphs in Section \ref{sec:variants}. \looseness=-1

There are two fundamental types of updates on graphs: \emph{edge insertion}, i.e., add an edge $(a, b)$ into $E$, and \emph{edge deletion}, i.e., delete an edge $(a,b)$ from $E$. Note that node insertion or deletion can be treated as a batch update containing only edge insertions or only edge deletions, respectively. A \emph{batch update} is a sequence of edge insertions and deletions. In the case that the same edge is being inserted and deleted within one batch update, we simply eliminate both of them. An update is \emph{valid} if it makes a change on a graph, i.e., inserting an edge $(a,b)$ into $G$ when $(a,b)\notin E$, and deleting an edge $(a,b)$ from $G$ when $(a,b) \in E$. We ignore invalid updates. 

Let $R\subseteq V$ be a subset of special vertices in $G$, called \emph{landmarks}. A \emph{label} $L(v)$ for each vertex $v \in V$ is a set of \emph{distance entries} $\{(r_i, \delta_L(r_i, v))\}^n_{i=1}$ where $r_i \in R$, $\delta_L(r_i, v) = d_G(r_i, v)$ and $n\leq |R|$. We call $(r_i, \delta_L(r_i, v))$ the \emph{$r_i$-label} of vertex $v$. \n{The set of labels for all vertices in $V$, i.e., $\{L(v)\}_{v \in V}$, form a \emph{distance labelling} over $G$. The \emph{size} of a distance labelling is defined as $\sum_{v\in V}|L(v)|$.} In the literature, a distance labelling is often constructed following the 2-hop cover property \cite{cohen2003reachability} \n{which requires at least one vertex $w \in L(u) \cap L(v)$ to be on a shortest path between $u$ and $v$.}
\begin{definition}[2-hop cover labelling] 
A distance labelling $L$ over $G=(V,E)$ is a \emph{2-hop cover labeling} if for any $s, t\in V$,
\label{eq:twoHop_query_distance}
\begin{align}
d_G(s, t) = \texttt{min}\{\delta_L(r_i, s) + \delta_L(r_i, t) |\notag\hspace{2cm}\\ (r_i, \delta_L(r_i, s)) \in  L(s), (r_i, \delta_L(r_i,t)) \in  L(t)\}.
\end{align}
\end{definition}

In our work, we consider a labelling property based on the notion of highway, i.e., highway cover labelling \cite{farhan2018highly}.  

\begin{definition}[Highway]
A \emph{highway} $H=(R, \delta_H)$ consists of a set $R$ of landmarks and a distance decoding function $\delta_H : R \times R \rightarrow \mathbb{N}^+$ s.t. $\delta_H(r_1, r_2) = d_G(r_1, r_2)$ for any two landmarks $r_1, r_2\ \in R$.
\end{definition}

\begin{definition}[Highway cover labelling]\label{def:highway-cover}
A \emph{highway cover labelling} $\Gamma=(H, L)$ consists of a highway $H$ and a distance labelling $L$ satisfying that, for any $v \in V\backslash R$ and $r\in R$,
\begin{align}\label{equ:highway-cover}
d_G(r, v) = \texttt{min}\{\delta_L(r_i, v) + \delta_H(r, r_i) |\notag\hspace{0cm}\\
(r_i, \delta_L(r_i, v)) \in  L(v)\}
\end{align}

\end{definition}
\n{Intuitively, a highway cover labelling requires that the label $L(v)$ of every vertex $v\in V$ must contain a distance entry to each landmark $r\in R$ unless there is another landmark on a shortest path between $r$ and $v$. Unlike a 2-hop cover labelling that can answer distance queries for any two vertices in a graph, i.e., \emph{a full distance labelling}, a highway cover labelling can only answer distance queries between any landmark and any vertex in a graph, i.e., \emph{a partial distance labelling}.}
\begin{definition}[Minimality]
A highway cover labelling $\Gamma=(H, L)$ over $G$ is \emph{minimal} if, for any highway cover labelling $\Gamma'=(H, L')$ over $G$, $size(L')\geq size(L)$ holds.
\end{definition} 

It has been shown in \cite{farhan2018highly} that for any fixed set of landmarks, there exists a unique minimal highway cover labelling, which is contained in every highway cover labelling.

%% file: section_Algorithms.tex
\section{Approach Overview}\label{sec:BDL}
\n{In this section, we present how to answer distance queries for any two vertices in a batch-dynamic graph by combining highway cover labelling with online searching.} 
The key idea is to dynamically maintain a highway cover labelling on a batch-dynamic graph, and then use such a highway cover labelling to bound online searches on a sparsified search space in order to accelerate query processing.

Given a highway cover labeling $\Gamma=(H, L)$, an upper bound on the distance between any pair of vertices $s, t \in V$ in a graph $G$ is computed as follows:
\begin{align}\label{eq:upper-distance-bound}
d^{\top}_{st} = \min \{&\delta_L(r_i, s) + \delta_H(r_i, r_j) + \delta_L(r_j, t) \mid \notag\\
& (r_i, \delta_L(r_i, s)) \in  L(s), (r_j, \delta_L(r_j, t)) \in  L(t)\} 
\end{align}

\n{Here, $d^{\top}_{st}$ is the minimal length amongst all paths between $s$ and $t$ that pass through the highway. Since there may exist a shorter path not passing through the highway, we conduct} a distance-bounded shortest-path search over a sparsified graph $G[V \backslash R]$ (i.e., removing all landmarks in $R$ from $G$) under the upper bound $d^{\top}_{st}$ \n{to answer the distance query $Q(s, t)$ such that }
\begin{align}
    Q(s,t)= \min (d_{G[V \backslash R]}(s,t),\; d^{\top}_{st} \notag)
\end{align}

\n{In the implementation, $d_{G[V \backslash R]}(s,t)$ can be computed by conducting a bidirectional BFS search from both $s$ and $t$ \cite{farhan2018highly} which terminates either after $d^{\top}_{st}-1$ steps or when the searches from both directions meet.}

The major challenge is: how to design an algorithm that can efficiently maintain a highway cover labelling for answering distance queries on graphs that undergo batch updates, particularly when graphs are very large?
\begin{algorithm}[t]
\caption{BatchHL (\textsc{BHL})}\label{algo:seq-algo}
\SetKwFunction{FBS}{BatchSearch}
\SetKwFunction{FBR}{BatchRepair}
\SetKwFunction{FMain}{BatchHL}
\SetKwProg{Fn}{Function}{}{end}
\Fn{\FMain{$G'$, $B$, $R$, $\Gamma$}}{
    $\Gamma'\gets\Gamma$ \\
    \ForEach{$r\in R$}{
        $\vaff\gets \FBS(G', B, r, \Gamma)$ \\
        $\FBR(G', \vaff, r, \Gamma, \Gamma')$ \\
    }
    \Return $\Gamma'$ \\
}
\end{algorithm}
\section{Proposed Method}
\label{sec:algorithm}
\n{In this section, we present our batch-dynamic method, BatchHL, which can efficiently maintain a minimal highway cover labelling for dynamic graphs. As described in Algorithm \ref{algo:seq-algo}, BatchHL involves two phases: \emph{Batch Search} and \emph{Batch Repair}.}

\subsection{Batch Search}
In the following let $G=(V,E)$ be a graph, $R\subseteq V$ a set of landmarks and $B$ a batch update resulting in the updated graph $G'=(V',E')$.
We denote the unique minimal highway labellings on $G$ and $G'$ by $\Gamma$ and $\Gamma'$, respectively. Our first aim is to identify vertices for which the set of shortest paths to a given landmark changes.

\begin{definition}[affected]\label{def:path-affected}
A vertex $v\in V$ is \emph{affected} by a batch update $B$ w.r.t. a landmark $r\in R$ iff $P_G(r,v)\neq P_{G'}(r,v)$.
\end{definition} 

We use $V_{\textsc{aff}}(r, B)=\{v\in V|P_{G}(v, r)\not=P_{G'}(v,r)\}$ to denote the set of all affected vertices by a batch update $B$ w.r.t. a landmark $r$. 
The following lemma states how affected vertices relate to a single update (either edge insertion or edge deletion).

\begin{lemma}\label{lemma:aff_vs}
A vertex $v$ is affected w.r.t. a landmark $r$ iff there exists a shortest path between $v$ and $r$ in $G\cup G'$ that passes through an inserted edge $(a, b)$ in $G'$ or a deleted edge $(a, b)$ in $G$.
\end{lemma}

An edge insertion or deletion $(a,b)$ can create or eliminate shortest paths starting from $r$ and passing through $(a,b)$. By this lemma, we know that any update on an edge $(a,b)$ with $d_{G}(r,a) = d_{G}(r,b)$ is \emph{trivial} w.r.t. a landmark $r$, since such an update does not affect any vertices w.r.t. the landmark $r$. 
\begin{algorithm}[t]
\caption{Batch Search}\label{algo:affected}
\SetCommentSty{textit}
\SetKwFunction{FMain}{BatchSearch}
\SetKwProg{Fn}{Function}{}{end}
\Fn{\FMain{$G'$, $B$, $r$, $\Gamma$}}{

    \ForEach{$(a,b)\in B$}{
        \If{$d_G(r,a) < d_G(r,b)$}{
            add $(d_G(r,a) + 1, b)$ to $\mathcal{Q}$ \\
        }
        \ElseIf{$d_G(r,a) > d_G(r,b)$}{
            add $(d_G(r,b) + 1, a)$ to $\mathcal{Q}$ \\
        }
    }
    \While{$\mathcal{Q}$ is not empty}{
        remove minimal $(d,v)$ from $\mathcal{Q}$ \\
        \If{$v\notin\vcpaff$}{
            add $v$ to $\vcpaff$ \\
            \ForEach{$w \in N_{G'}(v)$}{
                \If{$d+1\leq d_G(r,w)$}{\label{ln:affected-prune}
                    add $(d+1,w)$ to $\mathcal{Q}$ \\
                }
            }
        }
    }
    \KwSty{return} $\vcpaff$
}
\end{algorithm}
\begin{figure}[t!]

\centering
\begin{tikzpicture}
\foreach \pos/\name in {%
{(0,2)/r}, {(1,2)/a},
{(0,1)/b}, {(1,1)/c}, {(2,1)/d},
{(0,0)/e}, {(1,0)/f}, {(2,0)/g}}
    \node[vertex] (\name) at \pos {\name};
\foreach \source/\dest in {%
r/a, b/c, c/d, e/f, f/g}
    \path[edge] (\source) -- (\dest);
\foreach \name in {r}
    \draw (\name) circle (6pt);
\path[edge,dashed] (a) edge[left] node{+} (b);
\path[edge,dashed] (d) edge[below right] node[right=5pt]{+} (e);
\path[edge,dashed] (b) edge[right] node{-} (e);
\path[edge,dashed] (a) edge[right] node{-} (c);
\tikzcaption{(a)};
\end{tikzpicture}
\hspace*{0.2cm}
\begin{tikzpicture}
\foreach \pos/\name in {%
{(0,2)/r}, {(1,2)/a},
{(0,1)/b}, {(1,1)/c}, {(2,1)/d},
{(0,0)/e}, {(1,0)/f}, {(2,0)/g}}
    \node[vertex] (\name) at \pos {\name};
\foreach \source/\dest in {%
r/a, a/c, b/c, c/d, b/e, e/f, f/g}
    \path[edge] (\source) -- (\dest);
\foreach \name in {r}
    \draw (\name) circle (6pt);
\path[edge,dashed] (a) edge[left] node{+} (b);
\path[edge,dashed] (d) edge[below right] node[right=5pt]{+} (e);
\tikzcaption{(b)};
\end{tikzpicture}
\hspace*{0.2cm}
\begin{tikzpicture}
\foreach \pos/\name in {%
{(0,2)/r}, {(1,2)/a},
{(0,1)/b}, {(1,1)/c}, {(2,1)/d},
{(0,0)/e}, {(1,0)/f}, {(2,0)/g}}
    \node[vertex] (\name) at \pos {\name};
\foreach \source/\dest in {%
r/a, a/b, b/c, c/d, d/e, e/f, f/g}
    \path[edge] (\source) -- (\dest);
\foreach \name in {r}
    \draw (\name) circle (6pt);
\path[edge,dashed] (b) edge[right] node{-} (e);
\path[edge,dashed] (a) edge[right] node{-} (c);
\tikzcaption{(c)};
\end{tikzpicture}
\begin{align*}
\setlength{\arraycolsep}{1pt}
\begin{array}{c||rcc|c|c|c|c|c|c}
 \multirow{2}{*}{}&v&=&a & b & c & d & e & f & g  \\
\cline{2-10}
&d_G(r,v)&= &1 & 3 & 2 & 3 & 4 & 5 & 6 \\\hline
Anchor& d_{G'}(b,v)&=&1 & 0  & 1 & 2 & 3 & 4 & 5\\\cline{4-10}
b &Eq.~\ref{equ:pattern} &=&False &True&False&False&False&False&False\\\hline
Anchor &d_{G'}(c,v)&= & 2&1&0&1&2&3&4\\\cline{4-10}
c&Eq.~\ref{equ:pattern} &= &False&True&True&True&True&True&True\\\hline
Anchor &d_{G'}(e,v)&= & 4&3&2&1&0&1&2\\\cline{4-10}
e & Eq.~\ref{equ:pattern} &=&False&False&False&False&True&True&True\\\hline
\end{array}
\end{align*}\vspace{-0.3cm}
\caption{\n{Example graphs, where edges marked by $+$ are inserted and edges marked by $-$ are deleted.}}\label{fig:bs_examples}\vspace{-0.4cm}
\end{figure}

\n{A naive way of finding affected vertices would be to apply Definition~\ref{def:path-affected} directly, by computing the set of all shortest paths from a landmark to each vertex on $G$ and $G'$, respectively,} and comparing them.
However, the computational cost of this would be prohibitive, even for small graphs.
\n{Until now, the standard way of handling graph changes is to treat edge insertion and edge deletion separately, since they have opposite effects on a graph. A natural extension on batch updates would then be to devise an incremental algorithm for batch edge insertions and a decremental algorithm for batch edge deletions. However, for a batch update that contains both edge insertions and edge deletions, 
we would then need to split it into two sub-batches - one for edge insertions and the other for edge deletions, and apply incremental and decremental algorithms, respectively.
Thus, repeated computations across edge insertions and deletions cannot be eliminated because no interaction between edge insertion and deletion can be captured.}

\begin{example}\label{E:batch-search}
\n{Consider Figure~\ref{fig:bs_examples}.a with four updates. If handling edge insertions and deletions separately in two sub-batches as shown in Figure~\ref{fig:bs_examples}.b-\ref{fig:bs_examples}.c, insertions of $(a,b)$ and $(d,e)$ lead to affected vertices $\{b,e,f,g\}$, while deletions of $(a,c)$ and $(b,e)$ lead to affected vertices $\{c,d,e,f,g\}$. The traversal on edges $(e,f)$ and $(f,g)$ is repeated.}
\end{example}

\n{To overcome the aforementioned shortcomings, we propose an efficient algorithm that unifies edge insertions and deletions. The key idea is based on our observation of a ``\emph{shared pattern}'' that characterises affected vertices w.r.t. a landmark in a unified way for both edge insertions and edge deletions.} 

\n{Let $r\in R$ and $(a,b)\in B$. Here, $(a,b)$ is any update, i.e., either inserted or deleted edge. The \emph{anchor} of $(a,b)$ is either $a$ or $b$, whichever is further away from $r$, and the \emph{pre-anchor} of $(a,b)$ is a vertex in $\{a,b\}$ that is not the anchor. The \emph{anchor distance} of $(a,b)$ is defined as $d_G(r, u')+1$ where $u'$ is the pre-anchor of $(a,b)$. Note that when $d_G(r, a)=d_G(r, b)$, there is no anchor nor pre-anchor corresponding to the update $(a,b)$. For each $B$, there exists a set of anchors corresponding to updates in $B$. An affected vertex $v$ in $G$ w.r.t. $r$ by a batch update $B$ can be found if the following condition is satisfied by at least one anchor $u$ from $B$:
}
\begin{equation}\label{equ:pattern}
 \n{d_G(r,v)\geq (d_G(r, u')+1)+d_{G'}(u, v).}  
\end{equation}

\begin{example}\label{E:batch-search2}
\n{Consider Figure~\ref{fig:bs_examples}.a again, which has three anchors $b$, $c$ and $e$ corresponding to the four updates. By applying Eq.~\ref{equ:pattern}, we can identify affected vertices $\{b,c,d,e,f,g\}$ as shown in the table.}\end{example}

\n{This striking pattern enables us to design a simple yet efficient algorithm for finding affected vertices which only needs to traverse local neighbors $v$ of each anchor $u$ on $G'$ recursively, i.e., computing $d_{G'}(u, v)$, regardless whether updates are edge insertions or deletions. The anchor distance $d_G(r, u')+1$ and the distance $d_G(r,v)$ on $G$ can be efficiently computed from the highway cover labelling $\Gamma$. The searches by different updates can be combined into a single search in the order of the anchor distances plus their distances to the anchors to avoid unnecessary computation.}  

\n{We note that due to this unified handling of insertions and deletions, optimization that apply to only one of these operations cannot simply be applied to the combined algorithm.
However, we show how one such optimization can still be leveraged in Section~\ref{sec:batch-search-improvements}.}

Armed with these ideas, Algorithm~\ref{algo:affected} eliminates unnecessary searches on unaffected vertices $v$ with $d_G(r,v)<d_G(r,u')+1$ and also avoids traversing vertices affected by multiple updates more than once.
However, Algorithm~\ref{algo:affected} does not precisely compute the set of all affected vertices, but a superset of it. The following example illustrates why this happens, and why it is difficult to avoid.

\begin{example}\label{E:detect-path-affected}
\n{Consider the graph in Figure~\ref{fig:more_examples}.a}.
The dotted edge between $r$ and $u$ indicates a long path between them, and the dotted edge between $r$ and $v$ indicates an even longer path.
When both edge deletion $(r,u)$ and edge insertion $(u,v)$ occur, the distance between $r$ and $u$ in $G$ is used to compute the anchor distance of $v$ for the update $(u,v)$, ignoring that the distance between $r$ and $u$ has changed. It is difficult to identify whether $v$ is affected -- it hinges on whether the long path between $r$ and $v$ is longer than the long path between $r$ and $u$ plus 1, which cannot be ascertained by $\Gamma$.
\end{example}

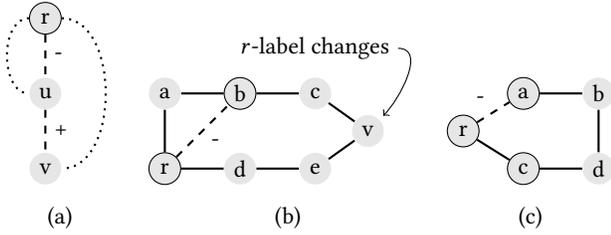
\begin{figure}[h]
\begin{tikzpicture}
\foreach \pos/\name in {%
{(0,2)/r},
{(0,1)/u},
{(0,0)/v}}
    \node[vertex] (\name) at \pos {\name};
\path[edge,dotted] (r) edge [out=-180,in=-180,left] (u);
\path[edge,dotted] (r) to [out=0,in=0] (v);
\foreach \name in {r}
    \draw (\name) circle (6pt);
\path[edge,dashed] (r) edge[right] node{-} (u);
\path[edge,dashed] (u) edge[right] node{+} (v);
\tikzcaption{(a)};
\end{tikzpicture}
\hspace{0.2cm}
\begin{tikzpicture}
\foreach \pos/\name in {%
{(0,1)/a}, {(1,1)/b}, {(2,1)/c},
{(2.7,0.5)/v},
{(0,0)/r}, {(1,0)/d}, {(2,0)/e}}
    \node[vertex] (\name) at \pos {\name};
\foreach \source/\dest in {%
r/a, r/d, a/b, b/c, c/v, d/e, e/v}
    \path[edge] (\source) -- (\dest);
\foreach \name in {r,b}
    \draw (\name) circle (6pt);
\path[edge,dashed] (r) edge[below right] node{-} (b);
\node (affected) at (2,1.6) {$r$-label changes};
\path[draw,->,shorten >= 2pt] (affected) to [out=0,in=45] (v);
\tikzcaption{(b)};
\end{tikzpicture}
\hspace{0.1cm}
\begin{tikzpicture}
\foreach \pos/\name in {%
{(0,1)/a}, {(1,1)/b},
{(-0.8,0.5)/r},
{(0,0)/c}, {(1,0)/d}}
    \node[vertex] (\name) at \pos {\name};
\foreach \source/\dest in {%
r/c, a/b, b/d, c/d}
    \path[edge] (\source) -- (\dest);
\foreach \name in {r,a,c}
    \draw (\name) circle (6pt);
\path[edge,dashed] (r) edge[above left] node{-} (a);
\tikzcaption{(c)};
\end{tikzpicture}\vspace{-0.3cm}\caption{\n{Example graphs for illustrating batch search.}}\label{fig:more_examples}\vspace{-0.4cm}
\end{figure}

We now characterize the set of vertices returned by Algorithm~\ref{algo:affected}.

\begin{definition}[composite path]\label{def:composite-path}
A path from $r$ to $v$ in $G\cup G'$ is a \emph{composite path} iff it consists of two parts: a part that lies in $G$ followed by a part in $G'$.
\end{definition}

A composite path is \emph{significant} iff it passes through at least one deleted and at least one inserted edge.
\n{In Figure~\ref{fig:bs_examples}.a, $r-a-b-c$ and $r-a-c-d$ are insignificant composite paths, $r-a-c-d-e$ is a significant composite path, and $r-a-b-e$ is not a composite path as a deleted edge comes after an inserted edge.}

\begin{definition}[composite-path affected]\label{def:composite-path-affected}
A vertex $v\in V$ is \emph{composite-path-affected} \n{\emph{(CP-affected)} by a batch update $B$ w.r.t. a landmark} $r\in R$ iff 
\begin{enumerate}
\item[(i)] $v$ is affected w.r.t. $r$, or
\item[(ii)] there exists a significant composite path from $r$ to $v$ of length $d_G(r,v)$ or less.
\end{enumerate}
\end{definition}

We will show that Algorithm~\ref{algo:affected} returns the set of all composite-path-affected vertices.
Clearly this includes all affected vertices.
Additional vertices due to condition (ii) are undesirable but hard to avoid, as illustrated in Example~\ref{E:detect-path-affected}.
From an algorithmic perspective, it happens because our starting distance is calculated w.r.t. $G$, so we are effectively considering paths for which the first part (from $r$ to an anchor) lies in $G$, and the rest in $G'$.

\begin{lemma}\label{L:batch-search-cpa}
Algorithm~\ref{algo:affected} returns the set of all \n{CP-affected} vertices.
\end{lemma}
\begin{proof}
We show that a vertex is \n{CP-affected} iff it lies in $\vcpaff$ returned by Algorithm~\ref{algo:affected}. We prove the “if” and “only if” below.

(if) Let $v\in\vcpaff$.
Then there must exist a composite path $p$ from $r$ to $v$ of length at most $d_G(r,v)$ that passes through at least one edge in $B$. If $p$ lies in $G$ then it lies in $P_G(r,v)$ but not in $P_{G'}(r,v)$, so $v$ is affected.
If $p$ lies in $G'$, then either it lies in $P_{G'}(r,v)$ or there exists an strictly shorter path $p'$ in $P_{G'}(r,v)$.
Neither $p$ nor $p'$ lies in $P_G(r,v)$, so $v$ is affected.
If $p$ lies neither in $G$ nor in $G'$ then it must be significant.
Thus $v$ is \n{CP-affected}. in all cases.

(only if) Reversely, let $v$ be CP-affected.
If $P_G(r,v)\not\subseteq P_{G'}(r,v)$ then there exists a path $p$ in $G$ of length $d_G(r,v)$ that passes through a deleted edge.
If $P_{G}(r,v)\subsetneq P_{G'}(r,v)$ then there exists a path $p$ in $G'$ of length at most $d_G(r,v)$ that passes through an inserted edge.
Otherwise the exists a significant composite path of length at most $d_G(r,v)$.
Thus, in all cases, there exists a composite path $p$ of length at most $d_G(r,v)$ that passes through an edge in $B$.

Let $(a,b)$ be either the last deleted edge that $p$ passes through, or the first inserted edge, with $d_G(r,a) < d_G(r,b)$.
Then $p$ can be split into $p_{ra}$ from $r$ to $a$, $(a,b)$ and $p_{bv}$ from $b$ to $v$ such that $p_{ra}$ lies in $G$ and $p_{bv}$ in $G'$.
The search in Algorithm~\ref{algo:affected} starting at $b$ will use $|p_{rb}| = d_G(r,a) + 1$ as the anchor distance for $b$, and proceed along $p_{bv}$.
Thus for every vertex $w\in p_{bv}$, including $v$, it will obtain $|p_{rw}| \leq d_G(r,w)$ and add $w$ to $\vcpaff$.
\end{proof}

\subsection{Improved Batch Search}\label{sec:batch-search-improvements}
So far we aimed at computing affected vertices.
However, changes to shortest paths between $r$ and $v$ do not always cause a change in distance. Thus we shall differentiate between new and eliminated paths, and strengthen the pruning condition $d+1\leq d_G(r,w)$ in Line~\ref{ln:affected-prune} of Algorithm~\ref{algo:affected} to $d+1 < d_G(r,w)$ for new paths.

Things get a little trickier though, as we may need to eliminate redundant labels, or restore previously eliminated labels when they become non-redundant.
Thus even if the distance between $r$ and $v$ does not change, the highway labeling may need to be updated.

\begin{example}\label{E:label-affected}
\newcommand{\vertices}{
\foreach \pos/\name in {%
{(0.8,2)/r},
{(0,1)/a}, {(1.6,1)/b},
{(0.8,0)/v}}
    \node[vertex] (\name) at \pos {\name};
}
Consider the following graphs and updates, where the landmarks are circled. In all cases, vertex $v$ is affected, but the distance between $r$ and $v$ does not change.
For case (a) adding the edge $(b,v)$ does not cause a label change for $v$.
It does however for case (b) where b is a landmark, causing the $r$-label of $v$ to be deleted.
Deletion of $(b,v)$ does not cause a change on the label of $v$ in case (c), but causes a change in case (d) where an $r$-label needs to be inserted.
\begin{center}
\begin{tikzpicture}[scale=0.8]
\vertices
\foreach \source/\dest in {%
r/a, r/b, a/v}
    \path[edge] (\source) -- (\dest);
\foreach \name in {r}
    \draw (\name) circle (8pt);
\path[edge,dashed] (b) edge[below right] node{+} (v);
\tikzcaption{(a) no change};
\end{tikzpicture}
\hfill
\begin{tikzpicture}[scale=0.8]
\vertices
\foreach \source/\dest in {%
r/a, r/b, a/v}
    \path[edge] (\source) -- (\dest);
\foreach \name in {r,b}
    \draw (\name) circle (8pt);
\path[edge,dashed] (b) edge[below right] node{+} (v);
\tikzcaption{(b) change};
\end{tikzpicture}
\hfill
\begin{tikzpicture}[scale=0.8]
\vertices
\foreach \source/\dest in {%
r/a, r/b, a/v}
    \path[edge] (\source) -- (\dest);
\foreach \name in {r,a}
    \draw (\name) circle (8pt);
\path[edge,dashed] (b) edge[below right] node{-} (v);
\tikzcaption{(c) no change};
\end{tikzpicture}
\hfill
\begin{tikzpicture}[scale=0.8]
\vertices
\foreach \source/\dest in {%
r/a, r/b, a/v}
    \path[edge] (\source) -- (\dest);
\foreach \name in {r,b}
    \draw (\name) circle (8pt);
\path[edge,dashed] (b) edge[below right] node{-} (v);
\tikzcaption{(d) change};
\end{tikzpicture}
\end{center}

\end{example}

A core difficulty in identifying whether affected vertices have changes on their labels is that label changes can happen far away from updates, and computing the changed labels of such vertices may require the consideration of vertices whose labels do not change, as illustrated by the example below.

\begin{example}\label{E:label-affected-isolated}
\n{Consider the graph in Figure~\ref{fig:more_examples}.b}, where $r$ and $b$ are landmarks and the edge $(r,b)$ is deleted.
The distance between $r$ and $c$ changes, but the label of $c$ does not change.
That is because the shortest path between $r$ and $c$ goes through landmark $b$ without change.
At the same time the label of $v$ does change, as the edge $(r,b)$ eliminates a shortest path between $r$ and $v$ that passes through landmark $b$, similar to case (d) in Example \ref{E:label-affected}. Although the label of $c$ does not change, the changed distance between $r$ and $c$ is needed for computing the changed label of $v$. Therefore, $c$ needs to be captured as well. 
\end{example}

We thus need to reexamine exactly which vertices need to be returned.
Firstly, this includes any vertex $v$ for which the highway labeling must be updated.
For non-landmarks the only possible change is to their $r$-label.
For landmarks their distance to $r$ is stored as part of the highway, and needs to be updated when it changes. Secondly, we must return any vertex for which the distance to $r$ changes.
That is because the batch repair algorithm computes the updated distance of a vertex to $r$ from that of its neighbors, so using outdated values for a neighbor could lead to errors.

\begin{example}\label{E:neighbor-distance}
\n{Consider the graph in Figure~\ref{fig:more_examples}.c}, where $r$, $a$ and $c$ are landmarks and the edge $(r,a)$ is getting deleted. The only node for which the highway labeling needs to be updated is $a$.
For $b$ the distance to $r$ changes, but its $r$-label is still redundant.
Using the old distance between $r$ and $b$ would cause our batch repair algorithm to compute $d_{G'}(r,a)$ as $d_{G}(r,b)+1=3$.
\end{example}

By considering vertices for which either label or distance changes, we can address both of the issues illustrated in Examples~\ref{E:label-affected-isolated} and~\ref{E:neighbor-distance}.
This motivates the following definition.

\begin{definition}[landmark-distance affected]\label{def:ld-affected}
A vertex $v$ is \n{\emph{landmark-distance-affected}} \emph{(LD-affected)} by a batch update $B$ w.r.t. a landmark $r\in R$ iff it is
\begin{enumerate}
\item[(i)]\emph{label-affected:} the $r$-label of $v$ changes, or
\item[(ii)]\emph{distance-affected:} the distance between $r$ and $v$ changes.
\end{enumerate}
\end{definition}

As seen in Example~\ref{E:label-affected}, changes to $r$-label without changes to distance happen whenever a new shortest path passing through another landmark is created where none existed previously, or when the last such path is deleted.
To identify such cases, we track whether a shortest path to $r$ passes through another landmark.

\begin{definition}[landmark length]
The \emph{landmark length} of a path $p$ starting from $r\in R$ is a tuple $(d,\lmf)\in\mathbb{N}\times\mathbb{B}$ where
\begin{itemize}
    \item $d$ is the length of $p$ (number of edges), and
    \item $\lmf$ is the \emph{landmark flag}, with $\lmf=\true$ iff $p$ passes through a landmark other than $r$.
\end{itemize}
We denoted this landmark length as $\lml{p}$.
The \emph{landmark distance} between $r$ and $v$ in $G$ is the minimal landmark length of paths between them, denoted as
\[
\dlm(r,v):=\min\big\{\; \lml{p} \mid \text{ $p$ is a path between $r$ and $v$ in $G$} \;\big\}
\]
The ordering used to compare landmark length tuples is the lexicographical one, with $\true < \false$.
The latter ensures that the landmark flag of $\dlm(r,v)$ is set iff \emph{any} of the shortest paths between $r$ and $v$ passes through another landmark.
\end{definition}

\begin{lemma}\label{L:r-label}
Let $\dlmp(r,v)=(d,\lmf)$.
If $d=\infty$ or $\lmf=\true$ then $v$ has no $r$-label in $\Gamma'$.
Otherwise $v$ has the $r$-label $(r,d)$.
\end{lemma}

\begin{proof}
If $v$ has any $r$-label in $\Gamma'$ it must be $(r,d)$.
As $\Gamma'$ is minimal, this $r$-label exists iff it is not redundant.
For $d=\infty$ redundancy of $(\infty, r)$ is obvious.
Otherwise $(d,r)$ is redundant iff the correct distance could also be computed using the highway.
This happens iff a shortest path between $r$ and $v$ passes through another landmark, which is indicated by the landmark flag.
\end{proof}

\begin{lemma}\label{L:ld-affected}
A vertex $v$ is LD-affected iff $\dlm(r,v)\neq\dlmp(r,v)$.
\end{lemma}

\begin{proof}
Let $\lmf_G$ and $\lmf_{G'}$ denote the landmark flags of $\dlm(r,v)$ and $\dlmp(r,v)$, respectively.
Condition (ii) of Definition~\ref{def:ld-affected} states $d_G(r,v)\neq d_{G'}(r,v)$.
It suffices to show that for $d_G(r,v)=d_{G'}(r,v)$ condition (i) holds iff $\lmf_G\neq\lmf_{G'}$. This is trivial for $d_G(r,v)=d_{G'}(r,v)=\infty$.
For finite distances it follow from Lemma~\ref{L:r-label}.
\end{proof}

Like Algorithm~\ref{algo:affected}, our improved batch search algorithm computes a superset of the set of all LD-affected vertices, albeit a smaller one.
By Lemma~\ref{L:ld-affected} we need to return a vertex whenever its landmark distance changes.
Thus we improve upon Algorithm~\ref{algo:affected} by tweaking the pruning conditions:
\begin{itemize}[leftmargin=*]
\item[--] Insertion: To affect the landmark distance, the landmark length of a new path $p_\text{new}$ from $r$ to $v$ must be strictly smaller than the current landmark distance between $r$ and $v$.
Thus we check $\lml{p_\text{new}}<\dlm(r,v)$.
\item[--] Deletion: A deleted path $p_\text{del}$ can only affect landmark distance if its landmark length was minimal, i.e., equal to the old landmark distance.
This suggests checking $\lml{p_\text{del}}=\dlm(r,v)$.
However, deleted paths may be obscured by shorter composite paths, so we check $\lml{p_\text{del}}\leq\dlm(r,v)$ instead.
\end{itemize}

The effects of these optimizations can be observed in Example~\ref{E:label-affected}, where $v$ will not be returned for case $(a)$ and case $(c)$.

To apply the new pruning conditions, we must know the landmark length of a path we are following, and whether or not it passes through a deleted edge.
Thus we track not only the length of each path, but also a landmark flag and a deletion flag.

\vspace{5pt}
\begin{definition}[extended landmark length]
The \emph{extended landmark length} of a path $p$ starting from $r\in R$ is a tuple $(d,\lmf,\delf)\in\mathbb{N}\times\mathbb{B}\times\mathbb{B}$ where
\begin{itemize}
\item $(d,\lmf)$ is the landmark length of $p$, and
\item $\delf$ is the \emph{deletion flag}, with $\delf=\true$ iff $p$ passes through a deleted edge.
\end{itemize}
We use lexicographical order for comparison, with $\true < \false$.
\end{definition}

For ease of extending landmark length values we will flatten tuples implicitly, i.e., we treat $((d,\lmf),\delf)$ as $(d,\lmf,\delf)$. The choice of the ordering $\true < \false$ for the deletion flag is not arbitrary.
When multiple search paths merge, we only track the length of the shorter one w.r.t. extended landmark length.
To ensure that deleted paths will not be pruned using the stricter condition for insertion, we need to keep the deletion flag if \emph{any} path has it, which is achieved by ordering $\true < \false$.

We apply our pruning conditions by comparing the extended landmark lengths computed for paths ending in $v$ to the landmark distance of $v$ in $G$.
For this we identify the minimal extended landmark length that indicates LD-affectedness.

\begin{lemma}\label{L:dmax}
Let $v$ be LD-affected w.r.t. $r$, and $\dmax$ defined as
\[
\dmax(r,v) :=
\big(\dlm(r,v),\; \true\big)
\]
Any composite path of minimal extended landmark length equals to $\dmax(r,v)$ or less and pass through an updated edge.
\end{lemma}

\begin{proof}
In the following we shall always refer to composite paths from $r$ to $v$.
By Lemma~\ref{L:ld-affected} we have $\dlm(r,v) \neq \dlmp(r,v)$.

(1) If $\dlm(r,v) < \dlmp(r,v)$ then all paths of minimal landmark length must pass through a deleted edge.
That makes their extended landmark length $\dmax(r,v)$ or less.

(2) If $\dlm(r,v) > \dlmp(r,v)$ then all paths  of minimal landmark length must pass through an inserted edge.
Their landmark length is at most $\dlmp(r,v)$, so their extended landmark length is strictly less than $\dmax(r,v)$.
\end{proof}
\begin{algorithm}[t]
\caption{Improved Batch Search}\label{algo:new-affected}
\SetCommentSty{textit}
\SetKwFunction{FMain}{BatchSearch}
\SetKwProg{Fn}{Function}{}{end}
\Fn{\FMain{$G'$, $B$, $r$, $\Gamma$}}{
    \ForEach{$(a,b)\in B$}{
        $\delf\gets (a,b)$ is deleted \\
        \If{$d_G(r,a) < d_G(r,b)$}{
            add $\big(\dlm(r,a)\oplus b,\; \delf,\; b\big)$ to $\mathcal{Q}$ \\
        }
        \ElseIf{$d_G(r,a) > d_G(r,b)$}{
            add $\big(\dlm(r,b)\oplus a,\; \delf,\; a\big)$ to $\mathcal{Q}$ \\
        }
    }
    \While{$\mathcal{Q}$ is not empty}{
        remove minimal $(d,\lmf,\delf,v)$ from $\mathcal{Q}$ \\
        \If{$v\notin\vcpaff$}{\label{ln:eld-merge}
            add $v$ to $\vcpaff$ \\
            \ForEach{$w \in N_{G'}(v)$}{
                $d_w\gets \big((d,\lmf)\oplus w,\; \delf\big)$ \\
                \If{$d_w\leq \dmax(r,w)$}{\label{ln:eld-prune}
                    add $(d_w,w)$ to $\mathcal{Q}$ \\
                }
            }
        }
    }
    \KwSty{return} $\vcpaff$
}
\end{algorithm}
Batch search with improved pruning is described in Algorithm~\ref{algo:new-affected}.
As we frequently need to update the landmark length of a path when appending another vertex, we define an operator for this:
\[
(d,\lmf) \oplus w :=
\begin{cases}
(d + 1, \true) & \text{if $w$ is a landmark} \\
(d + 1, \lmf) & \text{otherwise}
\end{cases}
\]

We finally show the correctness of Algorithm~\ref{algo:new-affected}, i.e., that all LD-affected vertices are included in its result set.
Note that some additional vertices may be returned as well.

\begin{lemma}\label{L:algo-eld-correct}
Algorithm~\ref{algo:new-affected} returns all LD-affected vertices.
\end{lemma}

\begin{proof}[Proof sketch]
Let $v$ be LD-affected, and $P_{\min}$ be the set of all composite paths from $r$ to $v$ of minimal landmark length.
By Lemma~\ref{L:dmax} these (and all their prefixes) meet the pruning condition in line~\ref{ln:eld-prune} and pass through an updated edge.
Thus the search in Algorithm~\ref{algo:new-affected} will follow them, starting from the last deleted or first inserted edge.
While some paths may be pruned in line~\ref{ln:eld-merge}, the search will still follow at least one path $p\in P_{\min}$ with minimal landmark length.
While its extended landmark length may not be minimal,
this only causes $p$ to be pruned if its landmark length equals $\dlm(r,v)$ and $p$ does not pass through a deleted edge.
But in this case, $v$ is not LD-affected.
\end{proof}

\subsection{Batch Repair}

In the following, we develop an efficient algorithm to repair labels.
At its core is an inference mechanism for the distances of affected vertices, which allows us to update their labels.
Here we start with \emph{boundary vertices} that lie on the boundary of affected and unaffected vertices, and for which the distance to $r$ can be computed from neighboring vertices whose distance did not change.
Importantly, even though a vertex may be affected by multiple edge updates in a batch, its $r$-label only needs to be updated once.

Let $v\in\vcpaff$.
For every neighbour $w$ of $v$ in $G'$, $d_{G'}(r,v)$ must be upper-bounded by $d_{G'}(r,w)+1$.
If such a neighbour lies outside of $\vcpaff$, the value of $d_{G'}(r,w)=d_G(r,w)$ can easily be obtained.
By taking the minimum of such known upper bounds, we get a readily available distance bound for $v$.
As we wish to eliminate redundant $r$-labels, we track landmark distance.

\newcommand*{\dlmb}{d^L_\textsc{bou}}
\newcommand*{\dbou}{d_\textsc{bou}}
\begin{definition}[Landmark distance bound]
Let $S\subset V\setminus\{r\}$ be a set of vertices. {The \emph{landmark distance bound} of $v$ w.r.t. $S$ is:}
\[
\dlmb(v,S) := \min \{ \dlmp(r,w)\oplus v \mid w \in N_{G'}(v)\setminus S \};
\]
{and the \emph{distance bound} of $v$ w.r.t. $S$ is:}
\[
\dbou(v,S) := \min \{ d_{G'}(r,w) + 1 \mid w \in N_{G'}(v)\setminus S \}.
\]
\end{definition}

Note that the distance bound of a vertex is simply the distance component of its landmark distance bound.
The following lemma allows us to compute the (landmark) distance of vertices in $\vcpaff$ from $r$ in $G'$ using their (landmark) distance bounds.

\begin{lemma}\label{L:min-distance-bound}
Let $S\subset V\setminus\{r\}$ and $v\in S$ with minimal distance bound.
Then $\dlmp(r,v)=\dlmb(v,S)$.
\end{lemma}

\begin{proof}
For $d_{G'}(r,v)=\infty$ this is trivial.
Otherwise let $p$ be a shortest path from $r$ to $v$ in $G'$ w.r.t. landmark length, $v'$ the first vertex in $p$ that lies in $S$, and $w$ its predecessor in $p$.
Since $w\notin S$ we have $\dlmb(v',S) \leq \dlmp(r,w) \oplus v = \dlmp(r,v')$.
If $v'\neq v$ then $d_{G'}(r,v') < d_{G'}(r,v) \leq \dbou(v,S)$, and therefore $\dbou(v',S) < \dbou(v,S)$.
This contradicts the minimality of $\dbou(v,S)$, so $v'=v$.
It follows that $\dlmb(v,S) = \dlmp(r,v)$.
\end{proof}

Note that $\dlmp(r,v)=\dlmb(r,S)$ does not generally hold for every boundary vertex $v$. This can e.g. be seen in the graph of Example~\ref{E:BatchHL-full}: when computing distances to $r_1$, $e$ must be repaired before $f$, as the new shortest path between $r_1$ and $f$ passes through $e$.
The situation is reversed when computing distance to $r_2$.
\newcommand*{\vardlmb}{D_\textsc{bou}}
\begin{algorithm}[t]
\SetCommentSty{textit}
\caption{Batch Repair}\label{algo:algo-update}
\SetKwFunction{FMain}{\textsc{BatchRepair}}
\SetKwProg{Fn}{Function}{}{end}
\Fn{\FMain{$G'$, $\vaff$, $r_i$, $\Gamma$, $\Gamma'$}}{
    \ForEach{$v\in\vaff$}{
        $\vardlmb[v]\gets \dlmb(v,\vaff)$ \label{ln:br-init}
        \tcp{use $\Gamma$ to compute}
    }
    \While{$\vaff$ is not empty}{
        $V_{\min}\gets\{ v\in\vaff \mid \vardlmb[v].d \text{ is minimal}\}$ \\
        remove $V_{\min}$ from $\vaff$ \\
        \ForEach{$v\in V_{\min}$}{
            \If{$\vardlmb[v].d=\infty \lor \vardlmb[v].\lmf$}{
                remove $r$-label from $\Gamma'(v)$ \\
            }
            \Else{
                set $r$-label of $\Gamma'(v)$ to $(r_i, \vardlmb[v].d)$ \\
            }
            \If{$v$ is a landmark}{
                $\delta'_H(r_i, v) \gets \vardlmb[v].d$ \\
            }
            \ForEach{$w \in N_{G'}(v)\cap\vaff$}{
                $\vardlmb[w]\gets\min(\vardlmb[w], \vardlmb[v]\oplus w)$ \\
            }
        }
    }
}
\end{algorithm}
Algorithm \ref{algo:algo-update} shows the pseudo-code of our batch repair algorithm.
Given a graph $G'$ and a set of all affected vertices $\vaff$, we first compute the landmark distance bounds of vertices in $\vaff$ using their unaffected neighbors.
We then find vertices in $\vaff$ with minimal distance bounds and remove them from $\vaff$.
By Lemma~\ref{L:min-distance-bound} their landmark distance to $r$ in $G'$ equals their landmark distance bounds.
We use these landmark distances to update their $r$-labels, as well as their highway distances in the case of landmarks.
Finally we update the landmark distance bounds of neighboring vertices in $\vaff$.
We continue this process until $\vaff$ is empty.

\subsection{Analysis of BatchHL}

In the following we will show correctness of Algorithm~\ref{algo:seq-algo} and analyse its time complexity.

\begin{theorem}
The highway labeling $\Gamma'$ returned by Algorithm~\ref{algo:seq-algo} is the minimal highway labeling for $G'$.
\end{theorem}

\begin{proof}
By Lemmas~\ref{L:batch-search-cpa} and~\ref{L:algo-eld-correct}, the vertex set $\vaff$ returned by \texttt{BatchSearch} contains all LD-affected vertices, regardless of which Algorithm (\ref{algo:affected} or~\ref{algo:new-affected}) is used. By Lemma~\ref{L:ld-affected} this means that for vertices outside of $\vaff$ the landmark distance to $r_i$ does not change, so that in line~\ref{ln:br-init} of Algorithm~\ref{algo:seq-algo} the value of $\dlmb(v,\vaff)$ can be computed from $\Gamma$.
From Lemma~\ref{L:min-distance-bound} it follows that $\vardlmb[v]=\dlmp(r_i,v)$ whenever vertex $v$ lies in $V_{\min}$.

For each landmark $r$ and each vertex LD-affected w.r.t. $r$ we update the $r$-label of $v$ in $\Gamma'$ based on its landmark distance to $r$ in $G'$.
By Lemma~\ref{L:r-label} these updates are correct.
As the $r$-labels of vertices outside of $\vaff$ do not change, and we initialized $\Gamma'$ using $\Gamma$, this leave all vertices with correct $r$-labels, for all $r\in R$, so the distance labeling of $\Gamma'$ is correct and minimal.
Highway is updated for vertices in $\vaff$ as well, for all $r\in R$, and do not change for others by Definition~\ref{def:ld-affected}.
\end{proof}

The following example illustrates the individual steps that our BatchHL algorithm runs through.

\begin{example}\label{E:BatchHL-full}
Consider the following graph and updates:
\begin{center}
\begin{tikzpicture}
\foreach \pos/\name/\lbl in {%
{(2,1.8)/a},
{(0,1)/b}, {(1,1)/r1/$r_1$}, {(2,1)/c}, {(3,1)/r2/$r_2$}, {(4,1)/d},
{(0,0)/e}, {(1,0)/f}, {(2,0)/g}, {(3,0)/h}, {(4,0)/i}}
    \node[vertex] (\name) at \pos {\lbl};
\foreach \source/\dest in {%
r1/a, r1/b, r1/c, r2/c, r2/d, r2/g,
b/e, d/i, f/g, g/h, h/i}
    \path[edge] (\source) -- (\dest);
\foreach \name in {r1,r2}
    \draw (\name) circle (6pt);
\path[edge,dashed] (r1) edge[right] node{-} (f);
\path[edge,dashed] (r2) edge[above right] node{+} (a);
\path[edge,dashed] (e) edge[above] node{+} (f);
\end{tikzpicture}
\end{center}
The initial highway labeling $\Gamma=(H,L)$ will look like this:
\begin{align*}
H &= \{ \delta_H(r_1,r_2) = 2 \},\\
L &=
\setlength{\arraycolsep}{1pt}
\begin{array}{c|c|c|c|c|c|c|c|c}
a & b & c & d & e & f & g & h & i \\
\hline
(r_1,1) & (r_1,1) & (r_1,1) &         & (r_1,2) & (r_1,1) & (r_1,2) & (r_1,3) & \\
        &         & (r_2,1) & (r_2,1) &         & (r_2,2) & (r_2,1) & (r_2,2) & (r_2,2)
\end{array}
\end{align*}
BatchHL will initialize $\Gamma'$ as $\Gamma$, and then run BatchSearch and BatchRepair for both $r_1$ and $r_2$.

For $r_1$ the basic BatchSearch described as Algorithm~\ref{algo:affected} returns
\[
\vcpaff=\{ r_2, d, e, f, g, h, i \}
\]
Here vertex $e$ is not actually affected, but still returned due to the composite path $r_1-f-e$.
Algorithm~\ref{algo:new-affected} returns only
\[
\vcpaff=\{ e, f, g, h \}
\]
For $r_2,d$ and $i$, the new paths through $a$ have the same landmark length as existing ones and are thus pruned.
The eliminated path $r_1-f-g-h-i$ has strictly greater landmark length than the existing path through $r_2$, and thus is pruned.
Note that $e$ is still returned due to the composite path $r_1-f-e$, despite not being LD-affected.

One of these sets is then used as input for BatchRepair, say $\vaff=\{e,f,g,h\}$.
The initial landmark bounds for this set are
\[
\dlmb(r_1,\ldots) =
\begin{array}{c|c|c|c}
e & f & g & h \\
\hline
(2,\false) & (\infty,\false) & (3,\true) & (5,\true)
\end{array}
\]
Here $e$ has the minimal distance bound, so we update $L(e)$ by setting its $r_1$-label to $2$ (which does not actually change $L'(e)$).
Afterwards $e$ is removed from $\vaff$ and the landmark bound of $f$ is updated to $(3,\false)$.
In the next iteration $f$ and $g$ are minimal, so the $r_1$-label in $L(f)$ is updated to $(r_1,3)$ and the $r_1$-label in $L'(g)$ is removed.
Finally $\dlmb(r_1,h)$ is updated to $(4,\true)$ and the $r_1$-label in $L'(h)$ is removed.
This leaves $L'$ as
\[
L' =
\setlength{\arraycolsep}{1pt}
\begin{array}{c|c|c|c|c|c|c|c|c}
a & b & c & d & e & f & g & h & i \\
\hline
(r_1,1) & (r_1,1) & (r_1,1) &         & (r_1,2) & (r_1,3) &         &         & \\
        &         & (r_2,1) & (r_2,1) &         & (r_2,2) & (r_2,1) & (r_2,2) & (r_2,2)
\end{array}
\]
Running BatchSearch for $r_2$ gives us one of
\begin{align*}
\vaff = \{ r_1, a, b, e \} \text{ or }
\vaff = \{ a, e \}
\end{align*}
depending on which algorithm (Algorithms~\ref{algo:affected} or ~\ref{algo:new-affected}) is used.
Running BatchRepair on either of those inserts $(r_2,1)$ into $L'(a)$ and $(r_2,2)$ into $L'(e)$
for the final updated highway labeling
\[
L' =
\setlength{\arraycolsep}{1pt}
\begin{array}{c|c|c|c|c|c|c|c|c}
a & b & c & d & e & f & g & h & i \\
\hline
(r_1,1) & (r_1,1) & (r_1,1) &         & (r_1,2) & (r_1,3) &         &         & \\
(r_2,1) &         & (r_2,1) & (r_2,1) & (r_2,3) & (r_2,2) & (r_2,1) & (r_2,2) & (r_2,2)
\end{array}
\]
\end{example}

\vspace{0.1cm}
\noindent\textbf{Complexity analysis.~}
\n{Table \ref{table:complexity} compares the time and space complexity of the state-of-the-art methods and our proposed method BatchHL for constructing and updating a distance labelling, and querying a distance. Let $a$ be the total number of affected vertices, $l$ be the maximum label size and $d$ be the maximum degree. In our method, $a$ refers to CP-affected vertices in the batch-update setting which is different from $\textsc{FulFD}$ \cite{hayashi2016fully} and $\textsc{FulPLL}$ \cite{d2019fully} in the single-update setting. We perform $|R|$ BFSs to construct highway labelling in $O(|R| \cdot |V|)$ time and space. Then, we update highway labelling in Algorithm \ref{algo:seq-algo} where Algorithm \ref{algo:affected} visits $O(a)$ vertices and for each affected vertex performs $d$ queries to check its affected neighbors in $O(d \cdot l)$ time. Thus, the time complexity of Algorithms \ref{algo:affected} and \ref{algo:new-affected} is $O(a \cdot d \cdot l)$. Note that Algorithm \ref{algo:new-affected} further reduces the total number of CP-affected vertices and is naturally faster than Algorithm \ref{algo:affected}. In practice, $l$ and $d$ are closer to the average values, and $a$ is usually orders of magnitudes smaller than the total number of vertices in a graph. Next, Algorithm \ref{algo:algo-update} repairs CP-affected vertices returned by Algorithm \ref{algo:affected} which in the worse case could repair the labels of all CP-affected vertices. To decide whether the label of an affected vertex needs to be repaired, we check its neighbors in $O(d)$. Thus, the time complexity of Algorithm \ref{algo:algo-update} is $(a \cdot d)$, and the overall time complexity of Algorithm \ref{algo:seq-algo} is $O(|R| \cdot a \cdot d \cdot l)$ using $O(V)$ space. We omit $l$ from the time complexity of Algorithm \ref{algo:algo-update} because we store distances for all unaffected neighbors of affected vertices in Algorithms \ref{algo:affected} and \ref{algo:new-affected}.} 

\n{$\textsc{FulFD}$ constructs a bit-parallel shortest-path tree for each $r \in R$ and 64 of its neighbors $N$ and requires $O(R \cdot N)$ time and space for every vertex $v \in V$ which is huge for a large network because $V$ could be very large. Similarly, the construction time and space of $\textsc{FulPLL}$ (PLL) and $\textsc{PSL}^*$ \cite{li2019scaling} (Parallel PLL using $t$ cores) is prohibitive and this is the reason why we do not have results on large datasets for these methods in our experiments. For state-of-the-art dynamic methods $\textsc{FulFD}$ and $\textsc{FulPLL}$, $a$ is the sum over all affected vertices by each update in a batch and is very large in practice because of unnecessary computations due to which they take longer to update a distance labelling.}  

\begin{table}[h!]
\centering
\caption{\n{Complexity analysis and comparison. }}\vspace{-0.3cm}
\label{table:complexity}
 \scalebox{0.8}{\begin{tabular}{| l || l l | l l | l l |}  \hline
    \multirow{2}{*}{Method}&\multicolumn{2}{c|}{Construction}&\multicolumn{2}{c|}{Update}&\multicolumn{2}{c|}{Query} \\\cline{2-7}
    & Time & Space & Time & Space & Time & Space \\
    
    \hline\hline
    BatchHL & $O(R\! \cdot\! V)$ & $O(R \!\cdot\! V)$ & $O(R\! \cdot\! a\! \cdot\! d\! \cdot\! l)$ & $O(V)$ & $O(E)$ & $O(V)$ \\
    \textsc{FulFD} & $O(R\! \cdot\! N \!\cdot\! V)$ & $O(R\! \cdot\! N\! \cdot\! V)$ & $O(R\! \cdot\! N \!\cdot\! a \!\cdot\! d \!\cdot\! l)$ & $O(V)$ & $O(E)$ & $O(V)$ \\
    \textsc{FulPLL} & $O(l^2 \!\cdot \!E)$ & $O(V\! \cdot\! l)$ & $O(a\! \cdot\! (E\! +\! V\! \cdot\! l))$ & $O(V)$ & $O(l)$ & - \\
    $\textsc{PSL}^*$ & $O(l^2\! \cdot\! E / t)$ & $O(V\! \cdot\! l)$ & - & - & $O(d\! \cdot\! l)$  & - \\
    BiBFS & - & - & - & - & $O(E)$ & $O(V)$ \\\hline
\end{tabular}}\vspace{-0.5cm}
\end{table}

%% file: section_ParallelLabelling.tex
\section{Variants}\label{sec:variants}

\noindent\textbf{Parallel batch updates.~}BatchHL can be parallelized at the landmark level. Let $\Gamma=(H,L)$ \n{be the unique minimal highway cover labelling over $G$.} 
\n{Then the unique minimal highway cover labelling $\Gamma'=(H',L')$ over $G'$} may differ from $\Gamma$ in: (1) \emph{highway:} $H$ is changed to $H'$; and (2) \emph{labels:} $L$ is changed to $L'$.

To enable the parallelism on highway, we store $H$ using a \emph{highway matrix} such that $h_{ij}=h_{ji}$ for each pair of landmarks $(r_i,r_j)$. Then, searches can be conducted in parallel to update the entries in this highway matrix. In the second case, for any vertex $v$, distance entries in $L(v)$ w.r.t. different landmarks are disjoint subsets. Thus updating distance entries in $L(v)$ w.r.t. different landmarks can be processed in parallel. Putting it all together, for any batch update, we run batch search and batch repair w.r.t. each landmark in parallel to speed up the performance. 

\medskip
\noindent\textbf{Directed and weighted graphs.~}
Our methods can be extended to directed and non-negative weighted graphs. For directed graphs, we use $d_G(s, t)$ to refer to the distance from vertex $s$ to vertex $t$ and store two sets of labels for each vertex $v$, forward $L_f(v)$ and backward $L_b(v)$ labels, containing pairs $(r_i, \delta_{r_iv})$ after performing forward and backward pruned BFSs w.r.t. every $r_i \in R$. Accordingly, we store forward $H_f = (R, \delta_{H_f})$ and backward highway $H_b = (R, \delta_{H_b})$, where for any two landmarks $\{r_i, r_j\} \in R$, $\delta_{H_f}(r_i, r_j) = d_G(r_i, r_j)$ and $\delta_{H_b}(r_i, r_j) = d_G(r_i, r_j)$. To repair the affected labels and highways affected by a batch update, we perform our batch search and batch repair methods twice: once in the forward direction and once in the backward direction. Then the upper bound for a distance query $(s, t)$ can be computed using $L_f(s)$, $L_b(t)$, $\delta_{H_f}$ and $\delta_{H_b}$ in the same way as described in Equation \ref{eq:upper-distance-bound}. \n{For weighted graphs, we can use pruned Dijkstra’s algorithm in place of pruned BFSs. We consider updates in the form of edge weight increase or decrease instead of edge insertion or deletion. 
Our methods can then handle weight increases in a similar way to edge deletions, and weight decreases in a similar way to edge insertions.}

%% file: section_Experiments.tex
\newcommand{\bhlpar}{\textsc{BHL}^\textsc{\hspace{-1pt}P}}

\section{Experiments}\label{sec:experiments}
We have implemented our algorithm to experimentally verify its efficiency and scalability on real-world large networks.

\vspace{-0.2cm}
\subsection{Experimental Setup}
In our experiments, all algorithms are implemented in C++11 and compiled with g++ 5.5.0 using the -O3 option. All the experiments are performed on a Linux server Intel Xeon W-2175 (2.50GHz CPU) with 28 cores and 512GB of main memory. 

\begin{table}[t]
  \centering
  \caption{Summary of datasets.}\vspace{-0.3cm}
  \label{table:datasets}
    \begin{tabular}{| l l | r r r r |} 
      \hline
      Dataset & Type & $|V|$ & $|E|$ & avg. deg & max. deg \\
      \hline\hline
      Youtube & social & 1.1M & 3M & 5.265 & 28754 \\
      Skitter & comp & 1.7M & 11M & 13.08 & 35455 \\
      Flickr & social & 1.7M & 16M & 18.13 & 27224 \\
      Wikitalk & comm & 2.4M & 5M & 3.890 & 100029 \\
      Hollywood & social & 1.1M & 114M & 98.91 & 11467 \\
      Orkut & social & 3.1M & 117M & 76.28 & 33313 \\
      Enwiki & social & 4.2M & 101M & 43.75 & 432260 \\
      Livejournal & social & 4.8M & 69M & 17.68 & 20333 \\
      Indochina & web & 7.4M & 194M & 40.73 & 256425 \\
      Twitter & social & 42M & 1.5B & 57.74 & 2997487 \\
      Friendster & social & 66M & 1.8B & 55.06 & 5214 \\ 
      UK & web & 106M & 3.7B & 62.77 & 979738 \\\hline
      \n{Italianwiki} & \n{social} & \n{1.2M} & \n{35M} & \n{33.25} & \n{81090} \\
      \n{Frenchwiki} & \n{social} & \n{2.2M} & \n{59M} & \n{26.36} & \n{137021} \\
      \hline
    \end{tabular}\vspace*{-0.4cm}
\end{table}

\vspace{0.1cm}
\noindent\textbf{Baseline methods.} We consider the following variants of our batch dynamic algorithm, (1) \textsc{BHL}: which uses the batch search described in Algorithm~\ref{algo:affected} and the batch repair described in Algorithm~\ref{algo:algo-update}, (2) $\textsc{BHL}^+$: which uses the improved batch search described in Algorithm~\ref{algo:new-affected} and the batch repair described in Algorithm~\ref{algo:algo-update}, and (3) $\textsc{BHL}^p$: which is a parallel variant of $\textsc{BHL}^+$. We compare these variants with the state-of-the-art methods as follows:
\begin{itemize}[leftmargin=*]
    \item[--] \textsc{FulFD} \cite{hayashi2016fully}: \n{A fully dynamic method that incorporates two algorithms \textsc{IncFD} and \textsc{DecFD} to update distance labelling for edge insertions and deletions, and then} combines it with a graph traversal algorithm to answer distance queries.
    \item[--] \textsc{FulPLL} \cite{d2019fully}: A fully dynamic 2-hop cover labelling method which is composed of two separate dynamic algorithms. The first algorithm was proposed in \cite{akiba2014dynamic} to answer distance queries on graphs undergoing edge insertions and the second algorithm was proposed in \cite{d2019fully} to answer distance queries on graphs undergoing edge deletions. This method is based on the pruned landmark labelling (PLL) \cite{akiba2013fast}.
    \item[--] \textsc{PSL*} \cite{li2019scaling}: A parallel algorithm which constructs pruned landmark labelling for static graphs to answer distance queries.
    \item[--] BiBFS \cite{hayashi2016fully}: An online bidirectional BFS algorithm which answers distance queries using an optimized strategy to expand searches from the direction with fewer vertices.
\end{itemize}
Note that \textsc{FulFD} and \textsc{FulPLL} can handle only a single edge insertion/deletion at a time. Thus, for a fair comparison, we also consider a unit-update variant of our algorithm: treating our method $\textsc{BHL}^+$ in the unit update setting by performing one update at a time. We call this unit-update variant as $\textsc{UHL}^+$. The code for \textsc{FulFD}, \textsc{FulPLL} and \textsc{PSL*} was provided by their authors and implemented in C++. We use the same parameter settings as suggested by their authors unless otherwise stated. \n{For a fair comparison, we also select high degree landmarks and set them to 20 in the same way as \textsc{FulFD} for our methods.} We set the number of threads to 20 for \textsc{PSL*} as well as for the parallel variant of our method $\textsc{BHL}^p$.

\vspace{0.1cm}
\noindent\textbf{Datasets.}~ \n{We use 14 large real-world networks from a variety of domains to verify the efficiency, scalability and robustness of our algorithm. Among them, Italianwiki and Frenchwiki are two real dynamic networks whose topology evolves over time.} We treat these networks as undirected and unweighted graphs, and their statistics are summarized in Table \ref{table:datasets}. They are accessible at Stanford Network Analysis Project \cite{leskovec2015snap}, Laboratory for Web Algorithmics \cite{BoVWFI}, \n{and the Koblenz Network Collection \cite{kunegis2013konect}.}

\begin{table*}[ht]
\centering
\caption{Comparing update time of our methods $\textsc{BHL}^+$, \textsc{BHL} and $\textsc{BHL}^p$ with the state-of-the-art dynamic methods, where the batch size is 1,000 and thus the update time reported for every method is for 1,000 updates. 
}\vspace{-0.3cm}
\label{table:updates_performance_2}
\resizebox{\textwidth}{\height}{
\begin{tabular}{| l || l l l l l l | l l l l l | l l l l l |}  \hline
    \multirow{2}{*}{Dataset}&\multicolumn{6}{c|}{Fully Dynamic Batch Update Time (sec.)}&\multicolumn{5}{c|}{Incremental Batch Update Time (sec.)}&\multicolumn{5}{c|}{Decremental Batch Update Time (sec.)} \\\cline{2-17}
    & $\textsc{BHL}^p$ \hspace{-0.5cm} & $\textsc{BHL}^+$ \hspace{-0.5cm} & \textsc{BHL} \hspace{-0.5cm} & $\textsc{UHL}^+$ \hspace{-0.5cm} & \textsc{FulFD} \hspace{-0.5cm} & \textsc{FulPLL} \hspace{-0.2cm} & $\textsc{BHL}^p$ \hspace{-0.5cm} & $\textsc{BHL}^+$ \hspace{-0.5cm} & $\textsc{UHL}^+$ \hspace{-0.5cm} \hspace{-0.5cm} & \textsc{IncFD} \hspace{-0.5cm} & \textsc{IncPLL}\hspace{-0.2cm} & $\textsc{BHL}^p$ \hspace{-0.5cm} & $\textsc{BHL}^+$ \hspace{-0.5cm} & $\textsc{UHL}^+$ \hspace{-0.5cm} & \textsc{DecFD \hspace{-0.5cm}} & \textsc{DecPLL}\hspace{-0.2cm} \\
    
    \hline\hline
    Youtube & 0.046 & 0.070 & 0.208 & 0.091 & 1.249 & \n{9.110} & 0.003 & 0.008 & 0.048 & 0.154 & 0.194 & 0.070 & 0.169 & 0.239 & 3.181 & \n{9.850} \\
    Skitter & 0.147 & 0.601 & 0.902 & 1.587 & 5.986 & \n{8.770} & 0.002 & 0.006 & 0.069 & 0.117 & 1.312 & 0.163 & 0.751 & 2.382 & 14.15 & \n{31.50} \\
    Flickr & 0.024 & 0.026 & 0.130 & 0.099 & 2.152 & \n{6.300} & 0.003 & 0.008 & 0.072  & 0.053 & 1.259 & 0.030 & 0.041 & 0.107 & 3.364 & \n{13.40} \\
    Wikitalk & 0.029 & 0.025 & 0.101 & 0.134 & 2.926 & \n{4.550} & 0.002 & 0.005 & 0.097 & 0.029 & 0.081 & 0.046 & 0.044 & 0.147 & 5.674 & \n{9.820} \\
    Hollywood & 0.008 & 0.014 & 0.115 & 0.056 & 4.423 & - & 0.001 & 0.002 & 0.046 & 0.090 & 27.53 & 0.017 & 0.031 & 0.071 & 8.401 & - \\
    Orkut & 0.537 & 1.775 & 5.855 & 4.539 & 13.30 & - & 0.005 & 0.014 & 0.127 & 0.367 & - & 0.677 & 0.035 & 5.921 & 23.94 & - \\
    Enwiki & 0.508 & 1.681 & 10.50 & 3.952 & 121.7 & - & 0.008 & 0.012 & 0.168 & 0.316 & 4.916 & 0.770 & 3.079 & 8.194 & 251.2 & - \\
    Livejournal & 0.221 & 0.306 & 0.873 & 0.379 & 4.736 & - & 0.006 & 0.010 & 0.202 & 0.244 & - & 0.299 & 0.570 & 0.731 & 4.736 & - \\
    Indochina & 0.543 & 1.181 & 1.547 & 9.575 & 20.63 & - & 0.015 & 0.011 & 0.308 & 0.141 & 4.680 & 0.553 & 1.346 & 19.20 & 44.92 & - \\ 
    Twitter & 13.29 & 49.62 & 115.7 & 125.6 & 5103 & - & 0.125 & 0.024 & 13.09 & 0.263 & - & 19.17 & 68.85 & 231.8 & 9460 & - \\
    Friendster & 0.409 & 0.410 & 0.811 & 21.93 & 23.27 & - & 0.163 & 0.035 & 20.96 & 0.254 & - & 0.420 & 0.738 & 21.87 & 30.38 & - \\
    UK & 14.45 & 41.46 & 40.79  & 56.50 & 110.1 & - & 0.218 & 0.055 & 4.349 & 0.258 & - & 14.99 & 42.29 & 75.20 & 257.3 & - \\\hline
    \n{Italianwiki} & \n{0.001} & \n{0.001} & \n{0.025}  & \n{0.051} & \n{6.623} & \n{-} & \n{-} & \n{-} & \n{-} & \n{-} & \n{-} & \n{-} & \n{-} & \n{-} & \n{-} & \n{-} \\
    \n{Frenchwiki} & \n{0.003} & \n{0.004} & \n{0.067}  & \n{0.098} & \n{5.289} & \n{-} & \n{-} & \n{-} & \n{-} & \n{-} & \n{-} & \n{-} & \n{-} & \n{-} & \n{-} & \n{-} \\\hline
\end{tabular}}\vspace{-0.3cm}
\end{table*}
\begin{figure}[ht]
\centering
\includegraphics[width=0.48\textwidth]{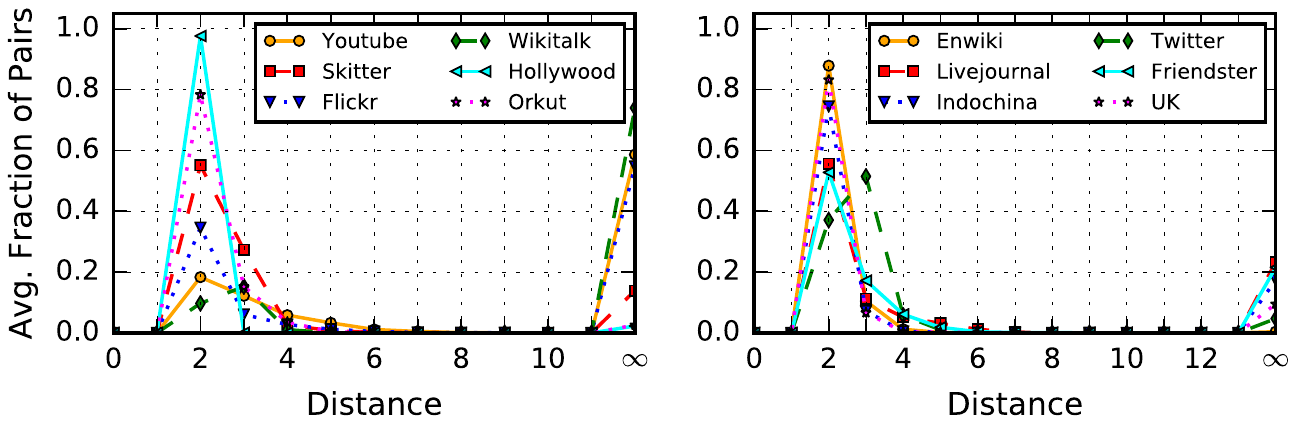} \vspace*{-0.8cm}
\caption{Distance distribution of batch updates.}
\label{fig:distance-distribution}\vspace{-0.4cm}
\end{figure}

\vspace{0.1cm}
\noindent\textbf{Test data generation.~}
For our batch dynamic variants, \n{we generate 10 batches for the first 12 datasets}, where each batch contains 1,000 edges randomly selected. We use three batch update settings for testing: (1) \emph{decremental} - delete these batches and measure the average deletion time, (2) \emph{incremental} - add these batches followed by decremental updates and measure the average insertion time, and (1) \emph{fully dynamic} - randomly select 50\% updates in each of these 10 batches to delete and then measure the average update time after applying these batches. \n{For the last two datasets, we select 10 batches in the order of their timestamps, each containing 1,000 real-world inserted/deleted edges and measure the average update time after applying them in a streaming fashion.}

For the methods \textsc{FulFD}, \textsc{FulPLL} and $\textsc{UHL}^+$, we randomly sample 1000 edges and follow the same update settings as above to measure the update time of performing updates one by one. These settings enable us to explore the impacts of edge insertions and edge deletions respectively, in addition to their combined impact. In Figure \ref{fig:distance-distribution}, we report the distance distribution of edges in these batches after deleting. As we can see, the distances in all datasets are small ranging from 1 to 6. This shows that the updates are mostly from densely connected components of these networks which may cause fewer vertices to be affected in the \emph{incremental} setting. Further, only a small number of updates are disconnected (i.e., have distance $\infty$) in most of these datasets. 

For queries, we randomly sample 100,000 pairs of vertices in each dataset to evaluate the average querying time on graphs being changed as a result of fully dynamic batch updates. We also report the average size of labelling in the fully dynamic setting.

\begin{table*}[ht]
\centering
\caption{Comparing performance of our method $\textsc{BHL}^+$ with the baseline methods in terms of construction time, query time and labelling size. Note that when a method did not finish the labelling construction in 24 hours, we denote it as ``-''. }\vspace{-0.3cm}
\label{table:updates_performance_1}
\begin{tabular}{| l || r r r r | l l l l | r r r r |}  \hline
    \multirow{2}{*}{Dataset}&\multicolumn{4}{c|}{Construction Time (CT) [s]}&\multicolumn{4}{c|}{Query Time (QT) [ms]}&\multicolumn{4}{c|}{Labelling Size (LS)} \\\cline{2-13}
    & $\textsc{BHL}^+$ & \textsc{FulFD} & \textsc{FulPLL} & \textsc{PSL*} & $\textsc{BHL}^+$ & \textsc{FulFD} & \textsc{FulPLL} & \textsc{PSL*} & $\textsc{BHL}^+$ & \textsc{FulFD} & \textsc{FulPLL} & \textsc{PSL*} \\
    
    \hline\hline
    Youtube & 2 & 4 & 84 & 4 & 0.005 & 0.010 & 0.045 & 0.002 & 20 MB & 83 MB & 3.14 GB & 318 MB \\
    Skitter & 3 & 8 & 511 & 21 & 0.029 & 0.020 & 0.082 & 0.007 & 42 MB & 153 MB & 11.9 GB & 1.01 GB \\
    Flickr & 3 & 10 & 546 & 23 & 0.007 & 0.013 & 0.102 & 0.005 & 34 MB & 152 MB & 13.1 GB & 0.98 GB \\
    Wikitalk & 2 & 5 & 92 & 4 & 0.006 & 0.008 & 0.031 & 0.001 & 41 MB & 74 MB & 5.22 GB & 160 MB \\
    Hollywood & 6 & 24 & 9,782 & 377 & 0.026 & 0.036 & - & 0.143 &  27 MB & 263 MB & - & 4.15 GB \\
    Orkut & 24 & 88 & - & 26,310 & 0.102 & 0.156 & - & 0.203 & 70 MB & 711 MB & - & 121 GB \\
    Enwiki & 25 & 88 & 7,382 & 389 & 0.053 & 0.051 & - & 0.021 & 82 MB & 608 MB & - & 7.04 GB \\
    Livejournal & 20 & 46 & - & 4,441 & 0.043 & 0.051 & - & 0.047 & 122 MB & 663 MB & - & 50.5 GB \\
    Indochina & 9 & 30 & 3,205 & 86 & 0.788 & 0.767 & - & 0.007 & 85 MB & 838 MB & - & 3.39 GB \\
    Twitter & 549 & 1,928 & - & - &0.868 & 0.174 & - & - & 1.14 GB & 3.83 GB & - & - \\
    Friendster & 1,181 & 3,365 & - & - & 0.815 & 0.902 & - & - & 2.43 GB & 9.14 GB & - & - \\
    UK & 178 & 621 & - & - & 1.174 & 5.233 & - & - & 1.78 GB & 11.8 GB & - & - \\\hline
    \n{Italianwiki} & \n{6} & \n{15} & - & \n{215} & \n{0.008} & \n{0.014} & - & \n{0.006} & \n{23 MB} & \n{159 MB} & - & \n{0.81 GB} \\
    \n{Frenchwiki} & \n{11} & \n{25} & - & \n{433} & \n{0.009} & \n{0.016} & - & \n{0.006} & \n{46 MB} & \n{272 MB} & - & \n{1.54 GB} \\
    \hline
\end{tabular}\vspace{-0.3cm}
\end{table*}
\begin{table*}[ht!]
 \centering
 \renewcommand{\arraystretch}{1.1} 
 \caption{Average number of vertices affected by BHL$^+$ and BHL after performing batch updates on all the datasets.}\vspace{-0.3cm}
 \label{table:avg_affected_vertices}
 \resizebox{\textwidth}{\height}{\begin{tabular}{| l | l || r r r r r r r r r r r r r r |}  \hline
    Method& Type & Youtube & Skitter & Flickr & Wikitalk & Hollywood & Orkut & Enwiki & Livejournal & Indochina & Twitter & Friendster & UK & \n{Italianwiki} & \n{Frenchwiki}  \\
    \hline\hline
    \multirow{3}{*}{}
    & Delete & 366 K	&	971 K	&	55 K	&	127 K	&	14 K	&	503 K	&	1,220 K	&	276 K	&	2,079 K	&	10,622 K	&	66 K	&	54,515 K & \n{-} & \n{-} \\
    $\textsc{BHL}^+$ & Add & 23 K	&	11 K	&	22 K	&	16 K	&	2 K	&	3 K	&	4 K	&	12 K	&	15 K	&	2 K	&	6 K	&	12 K & \n{-} & \n{-} \\
    & Mix & 166 K	&	834 K	&	42 K	&	81 K	&	7 K	&	293 K	&	712 K	&	156 K	&	200 K	&	8,341 K	&	36 K	&	54,026 K & \n{337} & \n{3 K} \\\hline
    \multirow{1}{*}{}
    \textsc{BHL} & Mix & 476 K	&	1,266 K	&	157 K	&	474 K	&	41 K	&	982 K	&	3,587 K	&	454 K	&	3,085 K	&	20,705 K	&	80 K	&	54,864 K & \n{9 K} & \n{45 K} \\\hline
 \end{tabular}}\vspace{-0.3cm}
\end{table*}

\vspace{-0.2cm}
\subsection{Performance Comparison}\label{subsec:Q1}

\subsubsection{Update Time}
Tables \ref{table:updates_performance_2} and \ref{table:updates_performance_1} show the average update time of our proposed and the baseline methods. 

\vspace{0.1cm}
\noindent\textbf{Fully dynamic setting.~}From Table \ref{table:updates_performance_2}, we see that our proposed methods $\textsc{BHL}^p$, $\textsc{BHL}^+$, and \textsc{BHL} significantly outperform \textsc{FulFD} and \textsc{FulPLL} on all datasets w.r.t. update time. In particular, our methods $\textsc{BHL}^p$ and $\textsc{BHL}^+$ are over 15 times faster than \textsc{FulFD} on most of the datasets and several orders of magnitude faster than \textsc{FulPLL}. \textsc{FulPLL} only works on four graphs and fails to scale to large graphs with more than 100 millions. Further, the performance difference of $\textsc{BHL}^+$ and \textsc{BHL} is due to the fact that our improved batch search in $\textsc{BHL}^+$ can further prune away affected vertices that do not need to be repaired, and in practice they are significant in amount as can be seen in Table \ref{table:avg_affected_vertices}. \n{Our methods also significantly outperform \textsc{FuLFD} on the real-world dynamic networks: Italianwiki and Frenchwiki.} We can also observe that the average update time of $\textsc{BHL}^+$, \textsc{BHL} and $\textsc{BHL}^p$ is always by far smaller than recomputing labelling from scratch, i.e., construction time of $\textsc{BHL}^+$ in Table \ref{table:updates_performance_1}.
Notice that, we consider the same construction time for $\textsc{BHL}^p$ and \textsc{BHL} as $\textsc{BHL}^+$, which is smaller than the construction time of baseline methods \textsc{FulFD} \cite{hayashi2016fully} and \textsc{PSL*} \cite{li2019scaling} on all datasets. We can see that the parallel variant of PLL (\textsc{PSL*}) still failed to construct labelling for the largest three datasets.

\vspace{0.1cm}
\noindent\textbf{Incremental setting.~}Table \ref{table:updates_performance_2} also shows that our methods $\textsc{BHL}^+$, $\textsc{BHL}^p$ are considerably faster than the baseline methods \textsc{IncFD} and \textsc{IncPLL}. Even though \textsc{IncFD} and \textsc{IncPLL} do not preserve the minimality of distance labellings and thus do not spend time to delete outdated and redundant label entries, they are still slower than our methods. We can also see $\textsc{BHL}^+$ and $\textsc{BHL}^p$ in the batch update setting are significantly faster than  $\textsc{UHL}^+$ in the unit update setting. This is because $\textsc{UHL}^+$ requires extra usage of resource for each single update and involves in repeated and unnecessary computations. Here it is also to note that $\textsc{BHL}^p$ does not perform well on the last four datasets as compared to \textsc{BHL}. This is because there only exist a very small number of average affected vertices against the total number of affected vertices as shown in Table \ref{table:avg_affected_vertices}. This confirms that the parallel variant of our method works very well when a large number of vertices are affected by batch updates; otherwise it may introduce unneeded thread overhead.

\vspace{0.1cm}
\noindent\textbf{Decremental setting.~}It is evident from Table \ref{table:updates_performance_2} that our methods $\textsc{BHL}^+$ and $\textsc{BHL}^p$ are much faster than \textsc{DecFD} and \textsc{DecPLL} on all the datasets in this setting. Especially, \textsc{BHL} and $\textsc{BHL}^p$ can achieve outstanding performance on networks which have a high average degree such as Twitter, Flickr and Hollywood. Due to inherent complexity of edge deletion on graphs (i.e., increasing distances), \textsc{DecFD} and \textsc{DecPLL} take very long in identifying and updating labels of affected vertices. As we can see, \textsc{DecPLL} does not have results on 8 out of 12 datasets. This is because while applying decremental updates their software either crashed or did not finish when the datasets are large that is why we don't have query time and labelling size after updates for these datasets in Table \ref{table:updates_performance_1}.  Furthermore, our methods $\textsc{BHL}^+$ and $\textsc{BHL}^p$ outperform $\textsc{UHL}^+$ because both leverage the benefit of handling updates in a batch and significantly reduce repeated computations during identifying and repairing the labels of affected vertices.

\subsubsection{Labelling Size}\label{subsection:ls}
Table \ref{table:updates_performance_1} shows that $\textsc{BHL}^+$ has significantly smaller labelling size than \textsc{FulFD}, \textsc{FulPLL} and \textsc{PSL*} on all the datasets. When an update occurs, the labelling size of \textsc{FulFD} remains unchanged because they store complete shortest-path trees at all times. In contrast, $\textsc{BHL}^+$ stores pruned shortest-path trees preserving the property of minimality. Nonetheless, the labelling size of $\textsc{BHL}^+$ remains stable in practice because the average label size is bounded by a constant, i.e., the number of landmarks. The labelling size of \textsc{FulPLL} may increase significantly because \textsc{IncPLL} does not remove outdated and redundant distance entries and there is also no bound on labelling size. The parallel variant of PLL (\textsc{PSL*}) which exploit PLL properties to reduce labelling size still produces labelling of very large size as compared to $\textsc{BHL}^+$.

\begin{figure*}[ht!]
\centering
\includegraphics[width=0.93\textwidth]{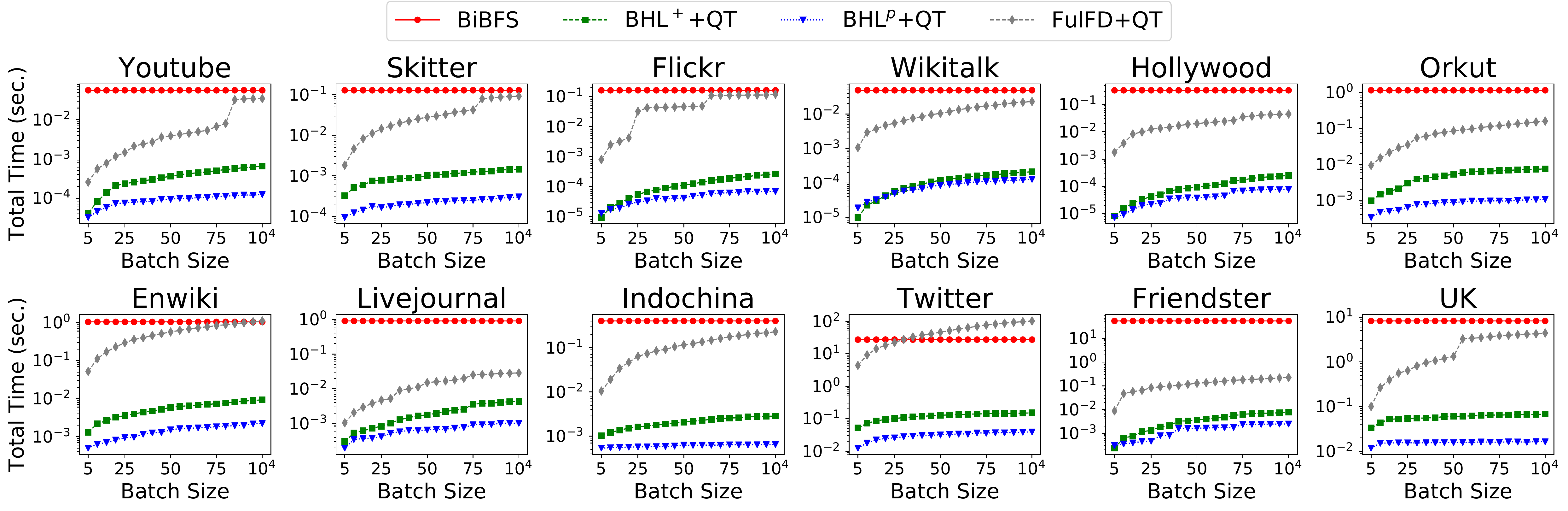}\vspace{-0.4cm}
\caption{\n{Total time of querying and updating by the proposed methods against online search methods.}}
\label{fig:scal}\vspace{-0.4cm}
\end{figure*}
\begin{figure}[ht]
\includegraphics[scale=0.31]{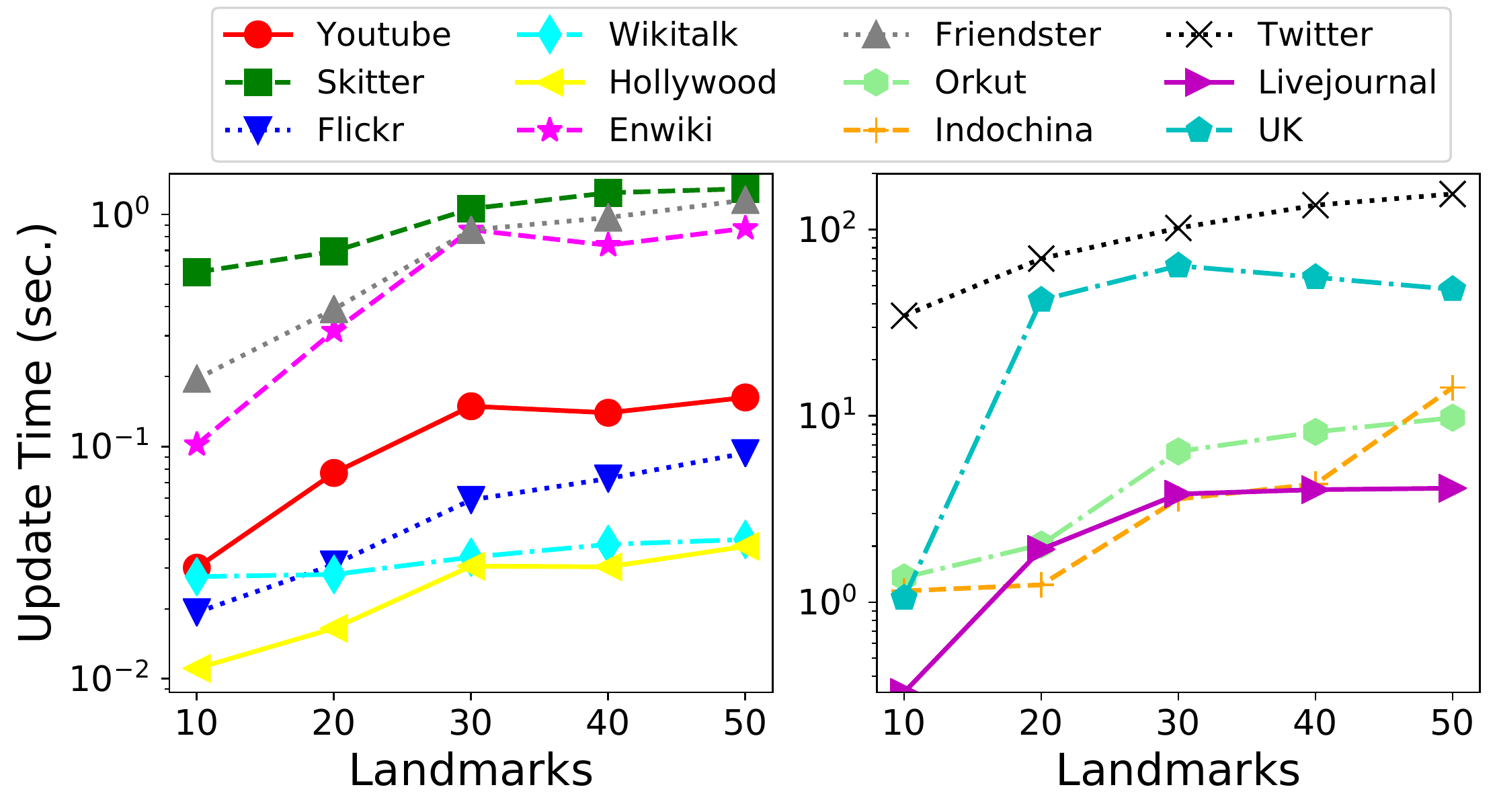}\vspace{-0.4cm}
\caption{\n{Update time under 10-50 landmarks.}}
\label{fig:update_time_varying_landmarks}
\includegraphics[scale=0.31]{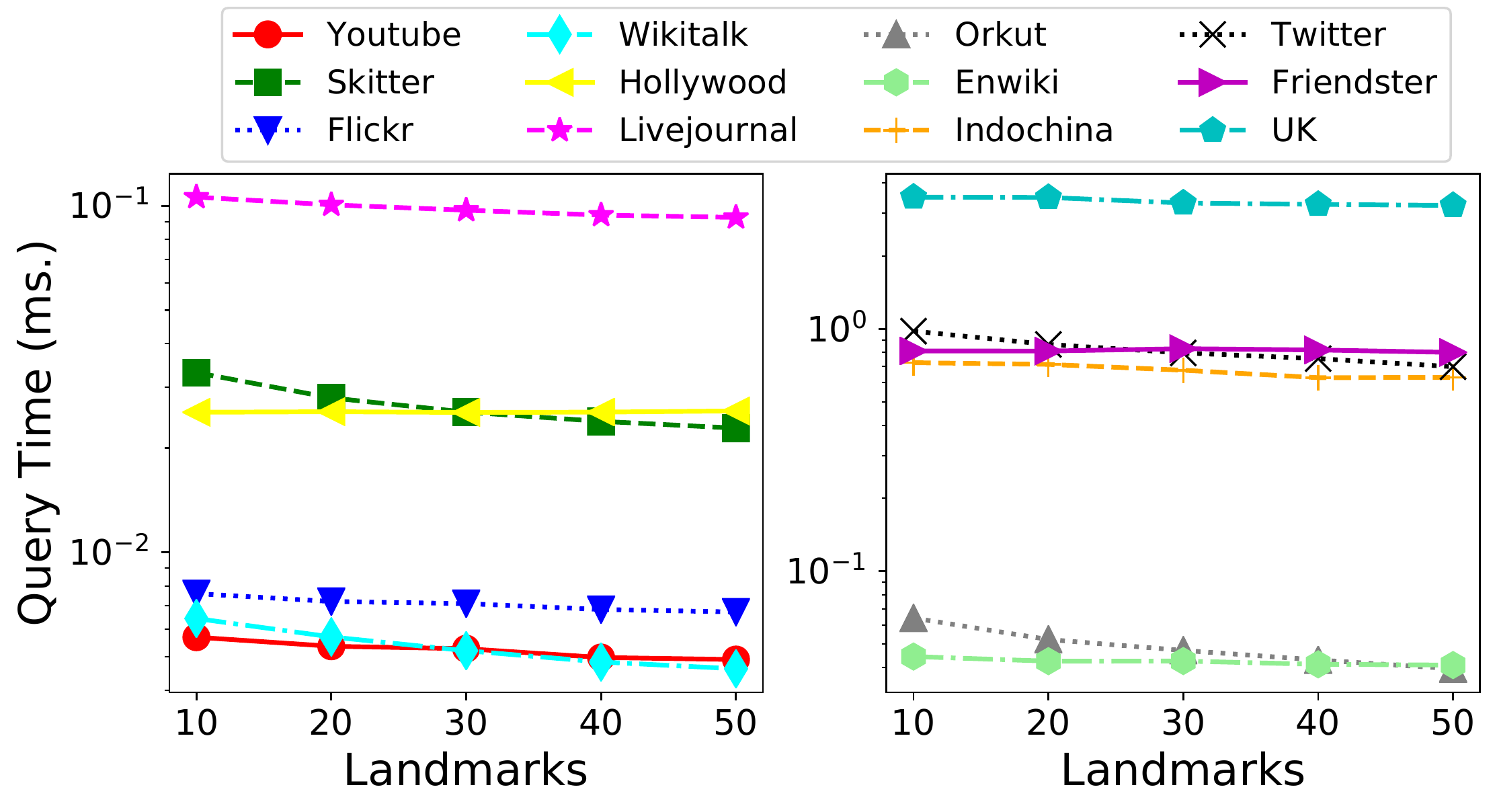}\vspace{-0.4cm}
\caption{\n{Query time under 10-50 landmarks.}}
\label{fig:query_time_varying_landmarks}\vspace*{-0.7cm}
\end{figure}

\subsubsection{Query Time}
Table \ref{table:updates_performance_1} shows that the average query time of $\textsc{BHL}^+$ is comparable with \textsc{FulFD} and faster than \textsc{FulPLL}. It has been previously shown \cite{d2019fully} that the average query time is largely dependent on labelling size. Since the dynamic operations do not considerably affect the labelling size for $\textsc{BHL}^+$ and \textsc{FulFD}, their query times remain stable. On Twitter, the query time of $\textsc{BHL}^+$ underperforms \textsc{FulFD} because \textsc{FulFD} also maintains the shortest-path information for the neighborhood of landmarks and we can see that the maximum degree of Twitter is very high which might cause many pairs to be covered by landmarks.  However, the query time for \textsc{FulPLL} may considerably increase over time because they do not remove outdated entries, leading to labelling of increasing sizes. Although the query time of \textsc{PSL*} in Table \ref{table:updates_performance_1} is better than $\textsc{BHL}^+$ on some datasets, it only handles static graphs. For dynamic graphs, it has the following limitations: (1) the cost of re-constructing labelling from scratch after each batch update is too high to afford, particularly when batch updates are frequent or when underlying dynamic graph is large which is evident from Table \ref{table:updates_performance_1}, (2) the labelling size is much larger than $\textsc{BHL}^+$. As we can see in Table \ref{table:updates_performance_1}, \textsc{PSL*} produces the labelling of size almost 99\% larger than the labelling of $\textsc{BHL}^+$ for Orkut thus possess a high query cost as well. Considering the overall performance w.r.t. three main factors i.e., query time, labelling size and construction time, $\textsc{BHL}^+$ stands out in claiming the best trade-off between these factors.

\vspace{-0.2cm}
\subsection{\n{Performance under Varying Landmarks}}\label{subsec:Q2}
\n{Figure \ref{fig:update_time_varying_landmarks} shows how the update time of our method $\textsc{BHL}^+$ in the fully dynamic setting behaves when increasing the number of landmarks. We can see that the update time for almost all datasets grows till 30 landmarks and then either decreases or remains stable. This is because selecting a larger number of landmarks can better leverage the pruning power of our method. On Twitter, we observe that the update time grows linearly due to its very high average degree which leads to a large fraction of vertices to be affected as can be seen in Table \ref{table:avg_affected_vertices} for $20$ landmarks. We can also see in Figure \ref{fig:query_time_varying_landmarks} that the query time decreases or remains the same for almost all datasets with the increased number of landmarks. Particularly, the query time of Twitter, Orkut and Livejournal decreases because they have a very high average degree and selecting a larger number of high degree landmarks contributes greatly towards shortest-path coverage and makes querying process faster.}

\vspace{-0.2cm}
\subsection{\n{Performance under Varying Batch Sizes}}\label{subsec:Q4}
\n{We also compare the total time of querying and updating on dynamic graphs. 
To make a fair comparison, the total time of our methods $\textsc{BHL}^+$ and $\textsc{BHL}^p$, and the baseline method $\textsc{FulFD}$ is the total time to perform a batch update plus the query time to perform 1000 queries after the batch update and then averaged over 1000 queries, denoted as $\textsc{BHL}^+$+QT, $\textsc{BHL}^p$+QT and $\textsc{FulFD}$+QT, respectively. We conduct the experiments for 5 randomly sampled fully dynamic batch updates of varying sizes, i.e., 500 to 10,000 in each batch. Figure \ref{fig:scal} presents the results. For the baseline method BiBFS, we take only the query time averaged over 1000 queries after applying a batch update. We see that, the overall performance of our methods is significantly better than the baseline methods on all the datasets. It is worth noticing that $\textsc{BHL}^p$ is not only more efficient than $\textsc{BHL}^+$, but also their efficiency gap becomes larger when the size of batch updates increases. This shows that the parallelism power of $\textsc{BHL}^p$ can be better leveraged for batch updates of larger sizes. We can also observe that the update time along with the query time of our methods grows fast for batches of smaller sizes (with up to 1000 updates) and then grows very slowly when batch sizes become very large which shows that our methods are robust w.r.t the increased batch size.}

\vspace{-0.2cm}
\subsection{\n{Performance on Directed Graphs}}
\n{We also conduct experiments on directed graphs. We can see in Table \ref{table:update_performance_d} that the update time of our methods is significantly smaller than the construction time of labelling from scratch. The update time of our optimized method $\textsc{BHL}^+$ is faster than the method $\textsc{BHL}$ on all datasets except Livejournal. On Livejournal, the amount of affected vertices traversed by both $\textsc{BHL}$ and $\textsc{BHL}^+$ is the same; however, due to additional overhead of computing extended landmark lengths, $\textsc{BHL}^+$ under-performs $\textsc{BHL}$. $\textsc{BHL}^p$ is still the fastest among all methods. Our methods are also efficient in performing queries and have small labelling sizes.} 

\vspace{-0.2cm}
\begin{table}[h!]
\centering
\caption{\n{Comparing update time, construction time (CT), query time (QT) and labelling size (LS) on directed graphs.}}\vspace{-0.3cm}
\label{table:update_performance_d}
\scalebox{0.85}{\begin{tabular}{| l || c c c | c c c |}  \hline
    Datasets & $\textsc{BHL}^p$[s] & $\textsc{BHL}^+$[s] & $\textsc{BHL}$[s] & CT[s] & QT[ms] & LS \\
    
    \hline\hline
    Wikitalk & 0.02 & 0.04 & 0.17 & 2.03 & 0.001 & 54 MB \\
    Enwiki & 2.98 & 12.5 & 28.0 & 46.8 & 0.023 & 177 MB \\
    Livejournal & 7.54 & 18.9 & 15.1 & 44.6 & 0.050 & 222 MB \\
    Twitter & 16.2 & 64.4 & 142 & 931 & 0.312 & 1.7 GB \\\hline
\end{tabular}}\vspace{-0.4cm}
\end{table}

%% file: section_Conclusion.tex
\section{Conclusion}\label{sec:conclusion}
We have proposed a novel method for answering distance queries on dynamic graphs undergoing batch updates. Our proposed approach exploits properties of updates in a batch to improve efficiency of maintaining distance labelling. We have analyzed the correctness and complexity of our approach and showed that they preserve the labelling minimality. We have empirically verified the efficiency and scalability of our approach on 14 real-world networks. For future work, we plan to explore the following directions: 1) the applicability and extension of the proposed method to road networks; 2) potential optimizations in designing separate batch-dynamic algorithms for edge insertion and edge deletion in dynamic graphs.